\documentclass[journal]{IEEEtran}
\IEEEoverridecommandlockouts

\usepackage[english]{babel}
\usepackage{amsmath,amsthm,amssymb}

\ifCLASSINFOpdf
  \usepackage[pdftex]{graphicx}					
  \DeclareGraphicsExtensions{.pdf,.jpeg,.png,.eps}
\else
  \usepackage[dvips]{graphicx}
  \DeclareGraphicsExtensions{.eps}
\fi
\usepackage[caption=false]{subfig}
\usepackage{ifpdf}
\ifpdf
  \usepackage{epstopdf}
\fi

\usepackage{cite,cleveref}
\usepackage{enumitem}
\usepackage{tikz}
\usetikzlibrary{decorations.pathreplacing,calc,shapes.geometric}

\definecolor{colorgreen}{HTML}{0F9D58} % blue
\definecolor{colorblue}{HTML}{0277BB}
\definecolor{colorred}{HTML}{DB4437}
\definecolor{coloryellow}{HTML}{F4B400}
\definecolor{colorviolet}{HTML}{4B0082}

\newcommand{\bstar}{\beta_{\mathrm{star}}}
\newcommand{\btree}{\beta_{\mathrm{tree}}}
\newcommand{\bline}{\beta_{\mathrm{serial}}}
\newcommand{\Bern}{\mathrm{Ber}}

\crefname{app}{Appendix}{Appendices}
\crefname{cor}{Corollary}{Corollaries}
\crefname{prop}{Proposition}{Propositions}
\crefname{lemma}{Lemma}{Lemmas}
\crefname{defn}{Definition}{Definitions}
\crefname{conj}{Conjecture}{Conjectures}
\crefname{exam}{Example}{Examples}
\crefname{supp}{Supplemental Section}{Supplemental Sections}

\newtheorem{theorem}{Theorem}
\newtheorem{cor}{Corollary}
\newtheorem{lemma}{Lemma}
\newtheorem{prop}{Proposition}

\newtheorem{defn}{Definition}

\newcommand{\bs}{\boldsymbol}
\newcommand{\bb}{\mathbb}
\newcommand{\mcal}{\mathcal}

\newcommand{\eye}{\bs{I}}
\newcommand{\zero}{\bs{0}}

\newcommand{\lb}{\left(}
\newcommand{\rb}{\right)}
\newcommand{\ls}{\left[}
\newcommand{\rs}{\right]}
\newcommand{\lc}{\left\{}
\newcommand{\rc}{\right\}}

\newcommand{\lv}{\left\vert}
\newcommand{\rv}{\right\vert}
\newcommand{\lV}{\left\Vert}
\newcommand{\rV}{\right\Vert}

\newcommand{\LRV}[1]{{\left\vert\kern-0.25ex\left\vert\kern-0.25ex\left\vert #1 \right\vert\kern-0.25ex\right\vert\kern-0.25ex\right\vert}}

\newcommand{\expect}[2]{\bb{E}_{#1}\lc#2\rc}

\newcommand{\nth}{^\mathsf{th}}

\newcommand{\matA}{\bs{A}}

\newcommand{\bbP}{\bb{P}}

\newcommand{\bbR}{\bb{R}}

\newcommand{\calO}{\mcal{O}}

\newcommand{\calT}{\mcal{T}}

\newcommand{\vecu}{\bs{u}}

\newcommand{\vecw}{\bs{w}}
\newcommand{\vecx}{\bs{x}}
\newcommand{\vecy}{\bs{y}}
\newcommand{\vecz}{\bs{z}}

\begin{document}

\title{Measurement Bounds for Compressed Sensing in Sensor Networks with Missing Data}
\author{Geethu Joseph and Pramod K. Varshney {\it Life Fellow, IEEE} 
\thanks{The authors are with the Dept.\ of EECS at Syracuse University, Syracuse, NY 13210, USA, emails:\{gjoseph,varshney\}@syr.edu.
This work was supported in part by National Science Foundation under Grant ENG 60064237. }
\thanks{This work is submitted in part to IEEE International Workshop on Signal Processing Advances in Wireless Communications, May 2020, Atlanta, GA, USA.}
}

\maketitle
\begin{abstract}
In this paper, we study the problem of sparse vector recovery at the fusion center of a sensor network from linear sensor measurements when there is missing data. In the presence of  missing data, the random sampling approach employed in compressed sensing is known to provide excellent reconstruction accuracy. However, when there is missing data, the theoretical guarantees associated with sparse recovery have not been well studied. Therefore, in this paper, we derive an upper bound on the minimum number of measurements required to ensure faithful recovery of a sparse signal when the generation of missing data is modeled using a Bernoulli erasure channel. We analyze three different network topologies, namely, star, (relay aided-)tree, and serial-star topologies. Our analysis establishes how the minimum required number of measurements for recovery scales with the network parameters, the properties of the measurement matrix, and the recovery algorithm. Finally, through numerical simulations, we show the variation of the minimum required number of measurements with different system parameters and validate our theoretical results. 
\end{abstract}

\section{Introduction}\label{sec:intro}
The signals of interest in several applications such as source localization, cognitive radio communication, MRI/ECG based health monitoring and anomaly detection in structural health monitoring are known to admit a sparse representation in an appropriate basis~\cite{meng2011collaborative,qaisar2013compressive,chae2013performance,joseph2016reconstruction,joseph2019anomaly}. While monitoring such physical phenomena using a wireless sensor network, a promising technique to collect data is compressed sensing (CS)-based data acquisition. In this data-acquisition technique, the compressed data from a sensor (which is referred to as a measurement henceforth) are linear projections of the unknown sparse signal. Thus, the sensor data sent to the fusion center (for example, the base station of the wireless sensor network) is a set of noisy linear measurements $\vecy\in\bbR^m$ of the unknown sparse vector $\vecx\in\bbR^N$:
\begin{equation}\label{eq:sys_star}
\vecy=\matA\vecx+\vecw,
\end{equation}
where $\matA\in\bbR^{m\times N}$ is the measurement matrix, and $\vecw\in\bbR^m$ is the bounded measurement noise, $\lv\vecw_i\rv\leq \sigma$ for $i=1,2,\ldots,m$. The fusion center recovers the sparse vector $\vecx$ using the standard CS algorithms like LASSO~\cite{tibshirani1996regression}, orthogonal matching pursuit~\cite{tropp2007signal}, or sparse Bayesian learning~\cite{tipping2001sparse} which faithfully recover $\vecx$ from $\vecy$, even if $m<N$. Therefore, the CS approach helps to enhance the overall lifetime of the sensor network by providing high reconstruction accuracy with low transmission overhead and less resource-demanding data compression method~\cite{bajwa2007joint,luo2010efficient,haupt2008compressed}.

CS-based recovery is a well-researched topic, both in terms of efficient reconstruction algorithms and strong theoretical guarantees. More importantly, CS-based data acquisition is also robust against missing data in the network that arise due to hardware failures, channel fading, synchronization issues, collisions, or environmental blockages~\cite{charbiwala2010compressive,cao2016data}. When there is missing data, only a portion of the measurement vector reaches the fusion center. We model this uncertainty in the received data using a channel model which we refer to as a \emph{Bernoulli erasure channel}. 
\begin{defn}[Bernoulli erasure channel\footnote{This channel model is different from the binary erasure channel that is widely used in information theory. In information theory, a binary erasure channel refers to a 2-input 3-output channel where there is an erasure output.}]\label{def:erasure_channel}
A data transmission channel is called a Bernoulli erasure channel with \emph{probability of observability} $p\in[0,1]$ if the receiving node (including the fusion center) receives the measurement (observes it) with probability $p$, and misses the measurement with probability $1-p$. The channel is represented using a Bernoulli distribution with mean equal to the probability of observability $p$.
\end{defn} 
Let $\calT\subseteq\{1,2,\ldots,m\}$  be the (known) random set of indices corresponding to the observed sensor measurements, i.e., the measurement vector available at the fusion center is given by
\begin{equation}\label{eq:sys_effective}
\vecy_{\calT} = \matA_{\calT}\vecx+\vecw_{\calT},
\end{equation}
where $\vecy_{\calT}\in\bbR^{\lv\calT\rv}$ denotes the measurements indexed by $\calT$, and the corresponding measurement matrix $\matA_{\calT}\in\bbR^{\lv\calT\rv\times N}$ is the submatrix of $\matA$ formed by the rows indexed by $\calT$. 

In our model formulation, missing data is handled by using random linear projections of the sparse signal, i.e., the rows of matrix $\matA$ are drawn independently from a common distribution. Since the random projections scatter the information contained in the network data over all (sensor) measurements, recovery is possible only if the number of measurements  $\lv\calT\rv$ that arrive at the fusion center is sufficiently large.\footnote{We discuss this point in detail in \Cref{sub:sparse_recovery}}. Thus, to ensure this, we need to increase the number of sensors in the network (oversampling) $m$. An important question in this context, which is also the focus of this paper, is as follows: \emph{How many additional sensors are required and how should they be arranged (network topology) to guarantee the same sparse signal recovery performance as in the case with no missing data when the measurements are obtained using random projections?} %The answer to this question depends on many factors such as the model for missing data mechanism, the topology of the sensor network, the recovery algorithm, the distribution of the sensing matrix, etc. 

A sufficient condition that guarantees exact sparse recovery from \eqref{eq:sys_effective} is obtained via an important property of the measurement matrix called the restricted isometric property (RIP):
\begin{defn}[Restricted isometric property]\label{def:RIP}
A  matrix $\matA\in\bbR^{m\times N}$ is said to satisfy the $s$-RIP with restricted isometric constant (RIC) $\delta_s$ if $\delta_s\in(0,1)$, where
\begin{multline}
\delta_s \triangleq \inf\bigg\{\delta: 1 - \delta \leq \lV\matA\vecz\rV^2 \leq 1 + \delta, \\
\forall \lV\vecz\rV=1 \text{ and } \lV\vecz\rV_0\leq s\bigg\},
\end{multline}
where $\lV\cdot\rV_0$ denotes the $\ell_0$ norm of a vector.
\end{defn}
The RIP of the measurement matrix determines the recovery quality of several CS algorithms like LASSO or basis pursuit (BP), compressive sampling matching pursuit (CoSaMP), or iterative hard thresholding (IHT)~\cite{foucart2013mathematical}. The RIP also ensures that the recovery process is robust to noise and is stable when the unknown vector is not precisely sparse. So our analysis relies on the RIP of the random measurement matrix $\matA_{\calT}$ whose entries as well as the dimension (the number of rows = $\lv\calT\rv$) are random. 

Before we present our specific models and the associated results, we first review the existing literature on the sparse signal recovery with missing data.

\subsection{Related literature}
CS-based recovery to handle missing data was introduced in  \cite{charbiwala2010compressive}. This work assumed that all the sensors are directly connected to the fusion center via independent Bernoulli erasure channels. Using numerical simulation, the authors of \cite{charbiwala2010compressive} showed that the recovery performance of CS techniques was on par with that of BCH codes. Following this work, several other researchers have also used the CS-based approach to missing data problems in different applications such as  network traffic reconstruction~\cite{roughan2011spatio}, localization refinement~\cite{rallapalli2010exploiting}, urban traffic sensing improvement~\cite{li2011compressive}, structural health monitoring~\cite{ji2014method,thadikemalla2017simple}, etc. Further, some other related problems like the design of novel sensing matrices using tight frames for robust data acquisition~\cite{chen2012benefit}, and spatial and temporal data loss models~\cite{kong2013data} have also been investigated.  Some works have also studied the problem of sparse signal recovery with missing data for other network topologies like tree topology and the serial-star topology~\cite{cao2016data}.

Most of the above works are limited in terms of the theoretical guarantees for signal recovery, and they mostly validate the results using numerical simulations tested on synthetic or real-world data-sets.  In \cite{charbiwala2010compressive}, it was established that $1/p$ times more measurements are required to ensure the same average probability of faithful recovery as that without missing data, where $p$ is the probability of observability of the Bernoulli erasure channel.\footnote{See \Cref{prop:existing} for the precise mathematical statement.} However, the classical CS recovery algorithms (with no missing data) come with much stronger assurances like conditions to ensure perfect recovery with high probability.  Motivated by this, we derive similar conditions that ensure that the probability of successful recovery is greater than the specified level when the missing data in every link of the network is modeled using a Bernoulli erasure channel.

\subsection{Our contributions}
In this work, we derive guarantees for the sparse signal recovery problem when there is missing data, given by the model in \eqref{eq:sys_effective}. Our assumptions are the following: 
\begin{itemize}
\item The measurements are obtained using a subGaussian random matrix  $\matA$.
\item All the links in the network are independent Bernoulli erasure channels (See \Cref{def:erasure_channel}). 
\end{itemize}
We consider three network topologies: a single hop star topology (parallel topology); a two-hop relay-aided tree topology; and a multi-hop serial-star topology. As the number of hops increases, the data are more likely be missing. The main contributions of the paper are as follows:
\begin{itemize}
\item \emph{Star topology:} We analyze, in \Cref{sec:star}, the RIP of the effective measurement matrix $\matA_{\calT}$ when the sensors are directly connected to the fusion center  through independent Bernoulli erasure channels.  The analysis (\Cref{thm:rip}), in turn,  characterizes the additional number of measurements required to handle the missing data uncertainty, with high probability (\Cref{cor:noisy}). We show that the number of measurements $m$ scales as $\lb \log \frac{1}{1-p+\exp(-C\delta^2)p}\rb^{-1}$ to ensure the exact recovery of the sparse vector. Here, $C$ and $\delta$ depend on the distribution of the measurement matrix and the recovery algorithm, respectively. 
\item \emph{Tree topology:} We study the conditions under which a sparse vector can be faithfully recovered in a two-hop tree network in \Cref{sec:tree}. The first hop is between the sensors and the relays, and the second hop is between the relays and the fusion center. We show that sparse recovery is successful with high probability if and only if the number of relays (not the number of sensors) exceeds a minimum value. We also derive a sufficient condition on the number of relays that drives the probability of successful recovery to an arbitrarily high value (\Cref{thm:riptree}).
\item \emph{Serial-star topology:} \Cref{sec:line} studies a serial-star topology which refers to the topology in which a number of serial multi-hop branches are connected to the fusion center in a star fashion. We show that, in this case, accurate recovery of the sparse vector with high probability happens only if the number of branches is large. Our result provides a bound on the number of branches to ensure this condition, which depends on the length of the branch, the probability of missing data in each link, the distribution of $\matA$, and the recovery algorithm (\Cref{thm:ripline}). 
\end{itemize}
In summary, our results provide insights on how the sparse recovery in a wireless sensor network depends on the parameters such as the number of sensors, branches or relays, and the probability of observability of each link. 

\textbf{Notation:} Boldface lowercase letters denote vectors, boldface uppercase letters denote matrices, and calligraphic letters denote sets. The rest of notation is given in \Cref{tab:notation}.
\begin{table}
\caption{Notation}
\label{tab:notation}
\begin{center}
{\normalsize
\begin{tabular}{|r|l|}
\hline
\textbf{Notation} & \textbf{Description}\\
\hline
$\bbR$  & Set of real numbers\\
$\eye$ & Identity matrix\\
$\zero$ & All zero matrix (or vector)\\
$\Vert \cdot \Vert$ & $\ell_2$ norm\\
$\Vert  \cdot \Vert_1$&  $\ell_1$ norm\\
$\lv\cdot\rv$ &Absolute value of a real number or the \\ 
& \hfill  cardinality of a set (depends on the context)\\
$\bbP\{\cdot\}$ & Probability of an event\\
$\Bern(\cdot)$ & Bernoulli distribution on $\{0,1\}$ \\
&\hfill parameterized by its mean\\
\hline
\end{tabular}
}
\end{center}
\end{table}

\section{Star Topology}\label{sec:star}

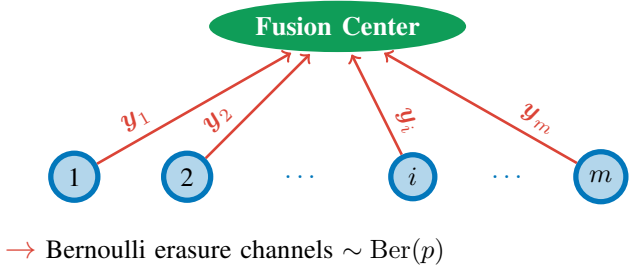
\begin{figure}[ht]
\begin{tikzpicture}

\tikzstyle{cloud} = [draw,  fill=colorblue!30,circle,text centered,scale = 1,line width=2pt,text=black,minimum height = 0.5cm]
\tikzstyle{box} = [fill=colorgreen,ellipse,  rounded corners, text centered,node distance = 1cm,line width=2pt]
\tikzstyle{line} = [->,line width=1pt]
\tikzstyle{comment} = [near start, above, sloped]

\node[box] 
  (I) {\textcolor{white}{\textbf{Fusion Center}}};
\color{colorblue}
  \path (I)+(-3.5,-2) node[cloud] (N1) {1};
  \path (I)+(-2,-2) node[cloud] (N2) {2};
  \path (I)+(1,-2) node[cloud] (N3) {$i$};
  \path (I)+(3.5,-2) node[cloud] (Nm) {$m$};

  \path (I)+(-0.5,-2) node (Nd1) {$\ldots$};
  \path (I)+(2.25,-2) node (Nd2) {$\ldots$};

\color{colorred}  
  \draw [line] (N1) -- (I) node [comment] (TextNode) {$\vecy_1$};
  \draw [line] (N2) -- (I) node [comment] (TextNode) {$\vecy_2$};  
  \draw [line] (N3) -- (I) node [comment] (TextNode) {$\vecy_i$};  
  \draw [line] (Nm) -- (I) node [comment] (TextNode) {$\vecy_m$};

\color{black}     
     \path (N1)+(2,-1) node (R) {\textcolor{colorred}{\large$\to$} Bernoulli erasure channels $\sim \Bern (p)$};
\end{tikzpicture}
\caption{Star topology with $m$ sensors.}
\label{fig:star}
\end{figure}
In this section, we consider a sensor network with star topology (also called as parallel topology) where $m$ sensors are individually connected to the fusion center as shown in \Cref{fig:star}~\cite{zafeiropoulos2009data,charbiwala2010compressive}. We assume that the channels between the sensors and the fusion center are unreliable\footnote{Here, we refer to the unreliablity due to missing data and do not refer to the channel noise. }and are independent of each other. The uncertainty introduced by the channel in terms of missing data  is modeled using an independent Bernoulli erasure channel with probability of observability as $p$ (See \Cref{def:erasure_channel}), i.e., every projection (entry of $\vecy$) is observed at the fusion center with probability $0\leq p\leq1$. The value of $p$ is independent of the missing and observed values and indices, i.e.,
\begin{equation}
\bbP\lc i\in\calT\rc=p, \hspace{0.5cm}i\in\{1,2,\ldots,m\},
\end{equation}
where $\calT\subseteq\{1,2,\ldots,m\}$ is the set of indices of the observed measurements. Our goal is to derive a sufficient condition on $m$ that ensures that the unknown vector $\vecx$ can be recovered from $\vecy_{\calT}$ with probability greater than a desired (high) value. 

We need the following definitions to present our results:
\begin{defn}[subGaussian random variable]
A random variable $A$ is said to be subGaussian with parameter $c$ if for any $\theta\in\bbR$,
\begin{equation}\label{eq:con_subgauss}
\expect{}{\exp\lb \theta A \rb } \leq \exp\lb c\theta^2 \rb.
\end{equation}
\end{defn}
\begin{defn}[subGaussian random matrix]\label{def:subGaussian}
A random matrix $\matA\in\bbR^{m\times N}$ is said to be a subGaussian random matrix with parameter $c$ if its entries are independent zero mean and unit variance subGaussian random variables with common parameter~$c$.
\end{defn}

The main result of the section is as follows:
\begin{theorem}\label{thm:rip}
Consider a star topological sensor network with $m$ sensors, whose measurement model is given by \eqref{eq:sys_star}. Each link in the network is modeled using an independent Bernoulli erasure channel with   probability of observability $p\in[0,1]$. Suppose $\matA\in\bbR^{m\times N}$  is a subGaussian random matrix with parameter $c$. Let $\calT\subseteq\{0,1,\ldots,m\}$ be the set of indices of measurements observed at the fusion center. Then, for any $\epsilon>0$, if $m\geq \bstar(p,\delta)$ where
\begin{multline}\label{eq:m_bound}
\bstar(p,\delta)\triangleq \ls\ln\frac{1}{ 1-p+p\exp(-C\delta^2)}\rs^{-1}\\
\times\ls \frac{4}{3}s\ln\lb \frac{eN}{s}\rb+\frac{14}{3}s+\frac{4}{3}\ln 2\epsilon^{-1}\rs,
\end{multline}
the RIC $\delta_s$ of $\lv\calT\rv^{-1}\matA_{\calT}$ given in \eqref{eq:sys_effective} satisfies $\delta_s < \delta$ for all $0<\delta<1$ with probability at least $1-\epsilon$. Here, $C>0$ is a constant dependent only on $c$. 
\end{theorem}
\begin{proof}
See \Cref{app:rip}.
\end{proof}
Next, we discuss some implications of the above result. 
\subsection{Sparse signal recovery guarantees}\label{sub:sparse_recovery}
An immediate consequence of the above result is the guarantee for sparse signal recovery with missing data.% presented in \Cref{sec:star}.
The RIP based analysis in \Cref{thm:rip} allows us to bound  the error in the recovery of the unknown vector $\vecx$ under bounded noise and model mismatch, i.e., when the measurements are noisy and the unknown vector is not exactly sparse, respectively.  
\begin{cor}\label{cor:noisy}
Consider a star topological sensor network with $m$ sensors, whose measurement model is given by \eqref{eq:sys_star}. Each link in the network is modeled using an independent Bernoulli erasure channel with   probability of observability $p\in[0,1]$. Suppose $\matA\in\bbR^{m\times N}$ is the measurement matrix of the network, which is a subGaussian random matrix with parameter $c$. Then, for some integer $r_{\mathrm{algo}}>0$ and  positive real number $\delta_{\mathrm{algo}}$, let  
\begin{multline}\label{eq:m_bound_noisy}
m \geq \ls\ln\frac{1}{ 1-p+p\exp(-C\delta_{\mathrm{algo}}^2)}\rs^{-1}\\
\times\ls \frac{4}{3}r_{\mathrm{algo}}s\ln\lb\frac{eN}{r_{\mathrm{algo}}s}\rb+\frac{14}{3}r_{\mathrm{algo}}s+\frac{4}{3}\ln 2\epsilon^{-1}\rs,
\end{multline} 
where $C>0$ is a constant dependent only on $c$. For any $\epsilon>0$, if \eqref{eq:m_bound_noisy} holds, with probability at least $1-\epsilon$, the unknown vector $\vecx$ can be recovered from the observed measurements given by \eqref{eq:sys_effective} with the following error bounds:
\begin{align}
\lV\vecx-\hat{\vecx}\rV_1 &\leq c_1\underset{\substack{\vecz\in\bbR^N\\\lV\vecz\rV_0\leq s}}{\min}\lV \vecx-\vecz\rV_1 + c_2\sqrt{s}\sigma\\
\lV\vecx -\hat{\vecx}\rV &\leq c_1\underset{\substack{\vecz\in\bbR^N\\\lV\vecz\rV_0\leq s}}{\min}\lV \vecx-\vecz\rV_1 /\sqrt{s} + c_2\sigma.
\end{align} 
Here, $\hat{\vecx}$ is the estimate of the unknown vector $\vecx$ determined by a sparse recovery algorithm, and $c_1,c_2>0$ are universal constants. 
The constants $r_{\mathrm{algo}}$ and $\delta_{\mathrm{algo}}$ are dependent on the recovery algorithms as follows:
\begin{itemize}
\item BP: $r_{\mathrm{algo}}=2$ and $\delta_{\mathrm{algo}} = \frac{4}{\sqrt{41}}$
\item 
IHT: $r_{\mathrm{algo}}=6$ and $\delta_{\mathrm{algo}}=\frac{1}{\sqrt{3}}$
\item 
CoSAMP: $r_{\mathrm{algo}}=8$ and $\delta_{\mathrm{algo}}=\frac{\sqrt{\sqrt{11/3}-1}}{2}$.
\end{itemize}
\end{cor}
\begin{proof}
We first note that the sparse signal recovery performance is independent of scaling. Since $\lv\vecw_i\rv<\sigma$, after scaling \eqref{eq:sys_effective} by $\lv\calT\rv^{-1}$, the norm of the noise term satisfies $\lV\lv\calT\rv^{-1}\vecw_{\calT}\rV\leq\sigma$. The result then follows from \Cref{thm:rip} and the upper bound on the RIC set by the different algorithms to ensure robust recovery~\cite[Theorems 6.12, 6.21, 6.28]{foucart2013mathematical}.
\end{proof}
The result reveals an elegant property of the sparse signal recovery guarantee. Assume that the number of measurements  $m$ is insufficient to recover an $s-$sparse unknown vector. Since the sufficient number of measurements is an increasing function of sparsity $s$, let $\tilde{s}<s$ be such that $m$ satisfies the bound in \eqref{eq:m_bound_noisy} for   $\tilde{s}$. Then, \Cref{cor:noisy} ensures that the algorithm recovers the $\tilde{s}-$sparse approximation of the unknown vector. 

\subsection{Special cases}
We consider the two extreme values for the probability $p$ with which a measurement is observed at the fusion center.
\begin{itemize}
\item \emph{$p=1$: }This value of $p$ corresponds to the case when there is no missing data, and \eqref{eq:m_bound} reduces to $m\geq \beta(\delta)$ where
\begin{equation}\label{eq:std_bound}
\beta(\delta) \triangleq \frac{1}{C\delta^2}\ls \frac{4}{3}s\ln\lb \frac{eN}{s}\rb+\frac{14}{3}s+\frac{4}{3}\ln 2\epsilon^{-1}\rs.
\end{equation}
As expected, the above bound is the same as the classical result on the sufficient number of measurements required for satisfying the RIC bound of a subGaussian matrix~\cite[Theorem 9.11]{foucart2013mathematical}.
\item \emph{$p=0$: }This value of $p$ corresponds to the case when no data are available to the recovery algorithm, and \eqref{eq:m_bound} reduces to $m\geq\infty$. This is justified as the recovery is not possible when no information about the unknown vector is available.
\end{itemize}

\subsection{Dependence on parameters}
The dependence of the bound $\bstar$ in \eqref{eq:m_bound}  on $N$ and $s$ is similar to the classical result on the sufficient number of measurements for a subGaussian matrix given in \eqref{eq:std_bound}. Both the bounds scale as $\calO(s\ln (N/s))$ if we keep $p$ and $\delta$ constant. Similarly, both the bounds are  $\calO(1/\delta^2)$. This is intuitive as a smaller value of $\delta$ ensures better recovery (see~\cite[Theorems 6.12, 6.21, 6.28]{foucart2013mathematical}) and thus, the number of measurements increases as $\delta$ decreases. 

We next study the influence of the probability of observability $p$. The ratio of the number of measurements required for recovery with and without missing data (ratio of the bounds in \eqref{eq:m_bound} and \eqref{eq:std_bound}) is
\begin{equation}\label{eq:ratio}
\frac{\bstar(p,\delta)}{\beta(\delta)} = C\delta^2\lb\ln\frac{1}{ 1-p+p\exp(-C\delta^2)}\rb^{-1}.
\end{equation}
Clearly, the ratio  depends only on three parameters: the probability of observability $p$, the parameter of the subGaussian matrix $c$, and the RIC value $\delta$. It is interesting to note that the ratio is independent of the dimension of the unknown vector $N$  and the sparsity $s$. The following proposition characterizes the properties of the bound $\bstar(p,\delta)$
\begin{prop}\label{prop:ratio}
Let $\bstar$ and $\beta$ be as defined in \eqref{eq:m_bound} and \eqref{eq:std_bound}. Then, the following statements hold:
\begin{enumerate}[label={(\roman*)}]
\item $\bstar(p,\delta)\geq \beta(\delta)/p$, for all $p\in[0,1]$ and values of $\delta\in(0,1)$.
\item $\bstar(p,\delta)$ is a strictly decreasing function of $p$ for a fixed value of $\delta\in(0,1)$.
\end{enumerate}
\end{prop}
\begin{proof}
See \Cref{app:ratio}.
\end{proof}
The above result is intuitive as $\bstar(p,\delta)\geq \beta(\delta)$ implies that more measurements are required when there is missing data. Also, as $p$ increases, more information about the unknown vector is available, and thus the bound $\bstar(p,\delta)$ decreases.
\subsection{Comparison with existing work}
The same model as that of ours, presented in \Cref{sec:star}, has been considered in \cite{charbiwala2010compressive}.  Their main theoretical result is as follows:
\begin{prop}\label{prop:existing}
Consider a star topological sensor network with $m$ sensors, whose measurement model is given by \eqref{eq:sys_star}. Each link in the network is modeled using an independent Bernoulli erasure channel with probability of observability $p\in[0,1]$. Suppose $\tilde{\matA}\in\bbR^{\tilde{m}\times N}$ is a random sensing matrix with independent and identically distributed rows, such that its RIC is $\delta$. Let $\matA\in\bbR^{m\times N}$ be the matrix formed by adding $m-\tilde{m}$ independent rows to $\tilde{\matA}$ which are distributed identically as the rows of $\tilde{\matA}$. If  $m=\frac{\bar{m}}{p}$, then the average  RIC of a suitably scaled version of $\matA_{\calT}$ is $\delta$ when averaged over the randomness in the missing data mechanism. Here, $\calT\subseteq\{0,1,\ldots,m\}$ is the set of indices of measurements observed at the fusion center. 
\end{prop}
The key observation obtained by comparing our results and \Cref{prop:existing} are as follows:
\begin{itemize}
\item \Cref{prop:existing} handles  the average probability for exact recovery, conditioned on the fact that $\bar{\matA}$ has RIC equal to $\delta$. Thus, the result deals only with the randomness in the missing data mechanism, and this probability only provides an insight into the average probability for faithful recovery. However, an explicit characterization of the average RIC of  a random matrix is not available in the literature. Consequently, the connection between the number of measurements and the other system parameters can not be established using this result. On the other hand, we provide a high probability result for successful recovery, which handles the randomness in both the measurement matrix and the missing data mechanism. Therefore, we derive an explicit relation between the number of measurements and the system parameters like sparsity $s$ and the length of the unknown vector $N$ (see \eqref{eq:m_bound} and \eqref{eq:m_bound_noisy}) and thus, our result is much stronger than \Cref{prop:existing} derived in~\cite{charbiwala2010compressive}.
\item Statement (i) of \Cref{prop:ratio} indicates that the bound $\bstar$ is greater than the bound given in \Cref{prop:existing}. This is expected as we need more number of measurements to ensure RIC with a high probability than those required to ensure an average RIC.
\end{itemize}

\section{Tree Topology}\label{sec:tree}
\begin{figure}[ht]
\begin{tikzpicture}
\tikzstyle{cloud} = [draw,  fill=colorblue!30,circle,text centered,scale = 1,line width=2pt,text=black,minimum height = 0.5cm]
\tikzstyle{box} = [fill=colorgreen,ellipse,  rounded corners, text centered,node distance = 1cm,line width=2pt]
\tikzstyle{dia} = [fill=coloryellow,rectangle,  rounded corners, text centered,node distance = 1cm,line width=2pt]
\tikzstyle{line} = [->,line width=1pt]
\tikzstyle{comment} = [midway, above, sloped]

\node[dia] 
  (I1) {\textcolor{white}{\textbf{Relay 1}}};
\color{colorblue}
  \path (I1)+(-1.25,-2) node[cloud] (N11) {1};
  \path (I1)+(-0.5,-2) node[cloud] (N12) {2};
  \path (I1)+(1,-2) node[cloud] (N1K) {$K$};
\color{colorred}    
  \path (I1)+(0.25,-2) node (Nd1) {$\ldots$};
 
  \draw [line] (N11) -- (I1);
  \draw [line] (N12) -- (I1);
  \draw [line] (N1K) -- (I1);

%*****************

  \path (I1)+(2.7,0) node[dia] 
  (I2) {\textcolor{white}{\textbf{Relay 2}}};
\color{colorblue}
  \path (I2)+(-0.75,-2) node[cloud] (N21) {1};
  \path (I2)+(0.75,-2) node[cloud] (N2K) {$K$};
  \path (I2)+(0,-2) node (Nd1) {$\ldots$};
\color{colorred}   
  \draw [line] (N21) -- (I2);
  \draw [line] (N2K) -- (I2);
  
%*****************
 \path (I2)+(1.5,0) node(C1) {\textcolor{coloryellow}{$\ldots$}};

  \path (I2)+(3.3,0) node[dia] 
  (I3) {\textcolor{white}{\textbf{Relay R}}};
\color{colorblue}
  \path (I3)+(-0.75,-2) node[cloud] (N31) {1};
  \path (I3)+(0.75,-2) node[cloud] (N3K) {$K$};
  \path (I3)+(0,-2) node (Nd1) {$\ldots$};
\color{colorred}   
  \draw [line] (N31) -- (I3);
  \draw [line] (N3K) -- (I3);

\path (I2)+(0.5,2) node[box] 
  (F) {\textcolor{white}{\textbf{Fusion Center}}};
  \color{colorviolet}
  \draw [line,dashed] (I1) -- (F); 
  \draw [line, dashed] (I2) -- (F);
  \draw [line,dashed] (I3) -- (F);
       
\color{black}
\path (N11)+(3,-1) node (R) {\textcolor{colorred}{\large$\to$} Bernoulli erasure channels $\sim \Bern (p)$};
\path (N11)+(3,-1.5) node (R) {\textcolor{colorviolet}{\large$\dashrightarrow$} Bernoulli erasure channels $\sim \Bern (q)$};
\end{tikzpicture}
\caption{Tree topology with $R$ relays and $K$ sensors connected to each of the relays.}
\label{fig:tree}
\end{figure}
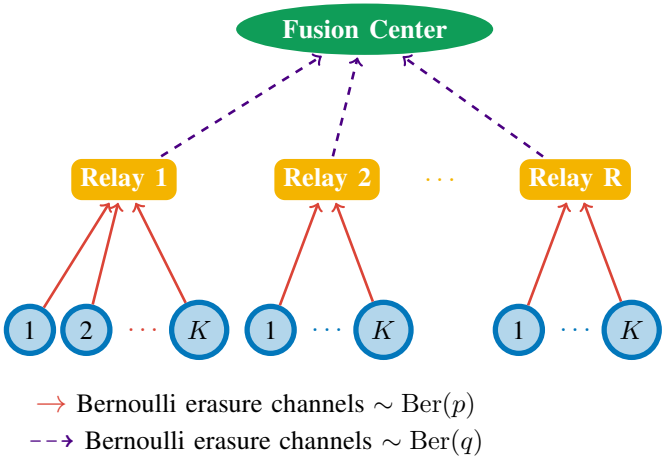
In this section, we consider a two-hop tree topology which consists of $R$ relays and $K$ sensors connected to each of these relays, as shown in \Cref{fig:tree}~\cite{cao2016data,zafeiropoulos2009data}. The measurement model is given by \eqref{eq:sys_star} where the total number of measurements $m=RK$. Each sensor sends its measurement to the relay to which it is connected, and these  measurements collected at the relays are sent to the fusion center by employing the `receive-and-forward' mechanism. All the channels, namely,  the $RK$ channels between the sensors and their corresponding relays, and the $R$ channels between the relays and the fusion center are assumed to be independent Bernoulli erasure channels (See \Cref{def:erasure_channel}). We assume that the probability of observability of the channel between a sensor and the corresponding relay node is $0\leq p\leq1$, and the probability of observability of the channel between a relay and the fusion center is $0\leq q\leq1$. Therefore, a sensor measurement is observable at the fusion center only if it is successfully transmitted from the sensor node to the relay node, and then from the relay node to the fusion center. Therefore, the probability of observing a sensor measurement at the fusion center is given by
\begin{equation}
\bbP\lc i\in\calT\rc=pq, \hspace{0.5cm}i\in\{1,2,\ldots,m\},
\end{equation}
where $\calT$ is the set of indices of sensor measurements observed at the fusion center. Suppose that the data transmitted from a relay to the fusion center is missing. Then, all the measurements corresponding to that relay are missing. Thus, unlike the star topology, in this case, observability of the measurements (corresponding to the same relay) are dependent. Consequently, the results of the star topology do not directly apply here.

Similar to  the star topology, the unknown vector $\vecx$ can be recovered using the observed sensor measurements given by \eqref{eq:sys_effective}. In the following, we derive the number of measurements required to ensure exact recovery of a sparse vector under the above model. Our main result is as follows:

\begin{theorem}\label{thm:riptree}
Consider a tree topological sensor network with $R$ relays and $K$ sensors per relay, whose measurement model is given by \eqref{eq:sys_star}. Each link in the network is modeled using an independent Bernoulli erasure channel with $p,q\in[0,1]$ as the probabilities of observability of a link between a sensor and a relay, and a relay and the fusion center, respectively. Suppose $\matA\in\bbR^{m\times N}$ is a subGaussian random matrix with parameter $c$.  Let $\calT\subseteq\{0,1,\ldots,m\}$ be the set of indices of measurements observed at the fusion center. Then, for any $\epsilon>0$, if $R\geq\btree(p,q,\delta,K)$ where
\begin{multline}\label{eq:m_tree}
\btree(p,q,\delta,K)\geq \ls\ln\frac{1}{ 1-q+q\lb 1-p+p e^{-C\delta^2}\rb^{K}}\rs^{-1}\\
\times\ls \frac{4}{3}s\ln\lb \frac{eN}{s}\rb+\frac{14}{3}s+\frac{4}{3}\ln 2\epsilon^{-1}\rs,
\end{multline}
the {RIC} $\delta_s$ of $\lv\calT\rv^{-1}\matA_{\calT}$ given in \eqref{eq:sys_effective} satisfies $\delta_s < \delta$ for all $0<\delta<1$ with probability at least $1-\epsilon$. Here, $C>0$ is a constant dependent only on $c$. 
\end{theorem}
\begin{proof}
See \Cref{app:riptree}.
\end{proof}

An immediate consequence of the above result is a corollary similar to \Cref{cor:noisy} which guarantees faithful recovery of a sparse vector when the bound $\btree$ in \eqref{eq:m_tree} is satisfied. We omit the statement of the corollary to avoid repetition. Also, as we have mentioned in \Cref{sub:sparse_recovery}, if $R$ is smaller than the bound $\btree$ in \eqref{eq:m_tree}, the result guarantees that the algorithm is guaranteed to recover a sparse approximation of the unknown vector. Another interesting corollary of \Cref{thm:riptree} is as follows:
\begin{cor}\label{cor:riptree}
Consider a more generalized setting of the star topology with $R$ sensors and the sensors have $K$ measurements each, whose measurement model is given by \eqref{eq:sys_star}. Each link in the network is modeled using an independent Bernoulli erasure channel with probability of observability $p$. Suppose $\matA$ is a subGaussian random matrix with parameter $c$. Let $\calT\subseteq\{0,1,\ldots,m\}$ be the set of indices of measurements observed at the fusion center. Then, for any $\epsilon>0$, if 
\begin{equation}
R\geq \bstar(p,\sqrt{K}\delta),
\end{equation}
the {RIC} $\delta_s$ of $\lv\calT\rv^{-1}\matA_{\calT}$ given in \eqref{eq:sys_effective} satisfies $\delta_s < \delta$ for all $0<\delta<1$ with probability at least $1-\epsilon$. Here, $C>0$ is a constant dependent only on $c$. 
\end{cor}
\begin{proof}
If we assume that $p=1$ for tree  topology setting in \Cref{thm:riptree}, the channels between the relay and the sensors in tree  topology are perfect. Then, the missing data are due to the uncertainty in the channel between the relay and the fusion center. This is equivalent to a star topology where each sensor has $K$ measurements, and the probability of observing data ($K-$length measurement vector) from a sensor is $q$. Thus, the result follows if we substitute $p=1$ and $q=p$ in \Cref{thm:riptree}.  
\end{proof}
We further note that if $K=1$, \Cref{cor:riptree} reduces to \Cref{thm:rip}. The corollary establishes a surprising relation between $\delta$ and the number of measurements per sensor $K$. The number of sensors required to ensure an RIC of $\delta$ for the star topology  with $K$ measurements per sensor is the same as those required to obtain an RIC of $\sqrt{K}\delta$ for the star topology with one measurement per sensor (note that $\bstar(p,\sqrt{K}\delta)< \bstar(p,\delta)$, for $K>1$). 

The other key insights from \Cref{thm:riptree} are presented in the following subsections:

\subsection{Nature of the bound} \label{sub:nature}
Unlike the star topology, the sufficient condition \eqref{eq:m_tree} in \Cref{thm:riptree} does not bound the number of measurements $m=RK$ directly. This result is expected due to the dependence of the observability of the measurements. Suppose that we increase $m$ by increasing $K$ while keeping $R$ constant. In that case, all the relays fail to send their measurements to the fusion center with probability $(1-q)^R$, i.e., none of the measurements reach the fusion center with probability $(1-q)^R$. Therefore, the probability of perfect recovery of the sparse vector can be driven to an arbitrarily large value only by increasing the number of relays $R$. As a result, the sufficient condition for exact recovery gives a lower bound on $R$. However, as we increase $R$, naturally the number of measurements $m=RK$ also increases.

Even though the bound \eqref{eq:m_tree} does not bound $K$, clearly the bound on $R$ diminishes with an increasing value of $K$ when the other parameters are kept constant. This trend is evident from the fact that $1-p+p e^{-C\delta^2}<1$, for all values of $p,\delta\in(0,1]$. The dependence is intuitive because as $K$ increases, more information is available to the system and consequently, the bound on the number of relays reduces.
 
Since the bound $\btree$ in \eqref{eq:m_tree} is a decreasing function of $K$, it also leads to a lower bound on $R$ which we obtain by taking $K\to\infty$:
\begin{equation}\label{eq:lowerbound_relay}
R\geq \ls\ln\frac{1}{ 1-q}\rs^{-1}\ls \frac{4}{3}s\ln\lb \frac{eN}{s}\rb+\frac{14}{3}s+\frac{4}{3}\ln 2\epsilon^{-1}\rs.
\end{equation}
The above bound is in agreement with our earlier point that there is a lower bound on $R$ that ensures perfect recovery of the sparse vector, irrespective of the value of $K$.
\subsection{Special cases}
We consider the following special cases:

\begin{itemize}
\item \emph{$p=1$ and $q=1$: } In this case, the bound $\btree$ in \eqref{eq:m_tree} simplifies as follows:
\begin{equation}
m=RK\geq \beta(\delta),
\end{equation}
where $\beta$ defined in \eqref{eq:std_bound} corresponds to the classical result which gives a bound on the number of measurements required for sparse signal recovery.
When $p=q=1$, none of the measurements are missing. Thus, this scenario is equivalent to having $RK$ measurements in the sparse recovery problem setting when there is no missing data. 
\item \emph{$q=1$: } For this setting, the bound $\btree$ in \eqref{eq:m_tree} reduces to
\begin{equation}
m=RK\geq \bstar(p,\delta),
\end{equation}
where $\bstar$ defined in \eqref{eq:m_bound}  corresponds to the star topology. When $q=1$, the channel between the relay and the fusion center is always reliable, i.e., all the sensors are connected to the  fusion center through a Bernoulli erasure channel with probability of observability being $p$. Therefore, this setting reduces to the star topology.
\item \emph{$p=0$ or $q=0$ or $K=0$: } When either of $p$, $q$ or $K$ is zero, the bound gives $m\geq \infty$. This is expected as in this setting, none of the measurements reach the fusion center and hence, the recovery is not possible.
\item \emph{$K=1$: } For this case, the bound $\btree$ in \eqref{eq:m_tree} simplifies as follows:
\begin{equation}
R\geq \bstar(pq,\delta).
\end{equation}
When $K=1$, the system is equivalent to having all the sensors connected to the fusion center via independent Bernoulli erasure channels with probability of observability as $pq$. 
\end{itemize} 
\subsection{Dependence on parameters}
Similar to the bound $\bstar$ in \eqref{eq:m_bound} for the star topology and the bound $\beta$ in \eqref{eq:m_bound} for sparse recovery without missing data, the bound $\btree$ in \eqref{eq:m_tree} scales $\approx \calO(s\ln (N/s)/\delta^2)$ if we keep the other parameters constant. We have already seen that $\btree(p,q,\delta,K)$ is a decreasing function of $K$ and $\delta$. The next result characterizes the dependence of the bound $\btree$ on the probabilities $p$ and $q$.
\begin{prop}\label{prop:ratiotree}
Let $\btree$,  $\bstar$ and $\beta$ be as defined in \eqref{eq:m_tree}, \eqref{eq:m_bound}, and \eqref{eq:std_bound}, respectively. Then, the following statements hold:
\begin{enumerate}[label={(\roman*)}]
\item $\btree(p,q,\delta,K)\geq \bstar(q,\sqrt{K}\delta)\geq\beta(\sqrt{K}\delta)/q$, for all $p\in[0,1]$ and $\delta\in(0,1)$, and all values of $K$.
\item $\btree(p,q,\delta,K)$ is a strictly decreasing function of $p$ and $q$ for a fixed value of $\delta$ and $K$.
\end{enumerate}
\end{prop}
\begin{proof}
Statement (i) follows from Statement (ii), the fact that $\bstar(1,q,\delta,K) = \beta(q,\sqrt{K}\delta)$ and \Cref{prop:ratio}. The proof of Statement (ii) is similar to the proof of \Cref{prop:ratio} given in \Cref{app:ratio}. So we omit the details here.
\end{proof}
Here, $\bstar(q,\sqrt{K}\delta)$ corresponds to the bound for the generalized star topology discussed in \Cref{cor:riptree}. Thus, Statement (i) of \Cref{prop:ratiotree} indicates that the recovery performance is better for the generalized star topology compared to the tree topology, when the number of measurements in the network $m=RK$ is the same. This is intuitive as the tree topology has a two-hop setting, and hence, the probability of missing data is higher. Also, Statement (ii) of \Cref{prop:ratiotree} is expected because as $p$ or $q$ decreases, more measurements are likely to be missing, and hence the bound $\btree$ increases.

\section{Serial-star topology}\label{sec:line}
\begin{figure}[ht]
\begin{tikzpicture}

\tikzstyle{cloud} = [draw,  fill=colorblue!30,circle,text centered,scale = 1,line width=2pt,text=black,minimum height = 0.5cm,node distance = 1.5cm]
\tikzstyle{box} = [ fill=colorgreen,ellipse,  rounded corners, text centered,node distance = 1cm,line width=2pt]
\tikzstyle{line} = [->,line width=1pt]
\tikzstyle{comment} = [midway, above, sloped]

\color{colorblue}
  \node[cloud] (N1K) {$K$};
  \node[below of = N1K,node distance = 1.5cm]   (N1d) {$\vdots$};
  \node[below of = N1d, cloud]   (N12) {$2$};
  \node[below of = N12, cloud]   (N11) {$1$};
  \node[below of = N11]   (N1) {Branch $1$};

\color{colorred}   
  \draw [line] (N1d) -- (N1K);
  \draw [line] (N12) -- (N1d);
  \draw [line] (N11) -- (N12);  

%*****************

\color{colorblue}
  \path (N1K)+(2,0) node[cloud] (N2K) {$K$};
  \node[below of = N2K,node distance = 1.5cm]   (N2d) {$\vdots$};
  \node[below of = N2d, cloud]   (N22) {$2$};
  \node[below of = N22, cloud]   (N21) {$1$};
  \node[below of = N21]   (N2) {Branch $2$};

\color{colorred}   
  \draw [line] (N2d) -- (N2K);
  \draw [line] (N22) -- (N2d);
  \draw [line] (N21) -- (N22);
  
%*****************
\color{colorblue}
  \path (N2K)+(2.5,0) node (N3K) {$\ldots$};
  \path (N2K)+(5,0) node[cloud] (N3K) {$K$};
  \node[below of = N3K,node distance = 1.5cm]   (N3d) {$\vdots$};
  \node[below of = N3d, cloud]   (N32) {$2$};
  \node[below of = N32, cloud]   (N31) {$1$};
  \node[below of = N31]   (N3) {Branch $R$};

\color{colorred}   
  \draw [line] (N3d) -- (N3K);
  \draw [line] (N32) -- (N3d);
  \draw [line] (N31) -- (N32);

\path (I2)+(0.5,2) node[box] 
  (F) {\textcolor{white}{\textbf{Fusion Center}}};
  \draw [line] (N1K) -- (F); 
  \draw [line] (N2K) -- (F);
  \draw [line] (N3K) -- (F);

 \draw[coloryellow,line width=1pt,dotted] ($(N1K)+(-0.75,0.75)$)  rectangle ($(N11)+(0.75,-1.5)$);
 \draw[coloryellow,line width=1pt,dotted] ($(N2K)+(-0.75,0.75)$)  rectangle ($(N21)+(0.75,-1.5)$);
 \draw[coloryellow,line width=1pt,dotted] ($(N3K)+(-0.75,0.75)$)  rectangle ($(N31)+(0.75,-1.5)$);  
       
\color{black}
\path (N11)+(2,-2) node (R) {\textcolor{colorred}{\large$\to$} Bernoulli erasure channels $\sim \Bern (p)$};
\end{tikzpicture}
\caption{Serial-star topology with $R$ branches and $K$ sensors in each branch.}
\label{fig:line}
\end{figure}
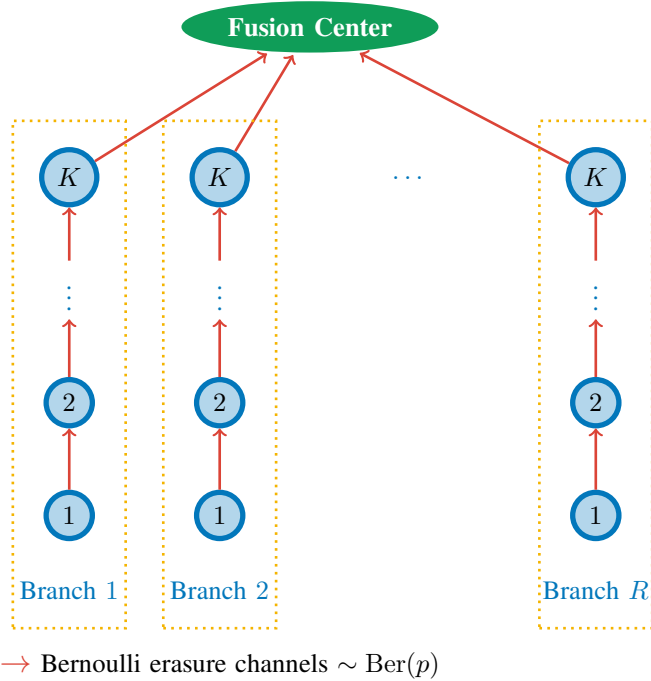
In this section, we consider a multi-hop serial-star topology with $m$ sensor nodes which are arranged in $R$ branches with each branch consists of $K$ hops, i.e., $m=RK$, as  shown in \Cref{fig:line}.  This topology is used when a large number of low-power sensors are deployed across a vast geographical area~\cite{zafeiropoulos2009data,cao2016data,luo2017serial}. When sensors are far away from the fusion center, the power consumed for reliable communication is high, and this affects the lifetime of the sensor network. So, the data is not directly sent from the far-away sensors to the fusion center; instead, it is sent serially through short-range and multi-hop communication.  We assume the conventional way of data aggregation where data are streamed in `sense-and-send' and `receive-and-forward' fashion. To be more specific, in the case of error-free  communication, the first sensor in a branch sends its measurement (for example, $\vecy_{1}\in\bbR$) to the second sensor. The second sensor adds its own measurement to the received measurement and sends (for example, $\begin{bmatrix} \vecy_1 &\vecy_2 \end{bmatrix}\in\bbR^2$) to the third sensor, and so on. Therefore, the $i\nth$ sensor in a branch receives $i-1$ measurements, and sends $i$ measurements to the $(i+1)\nth$ sensor. Finally, the last (or $K\nth$) sensors in all the branches send $K$ measurements each to the fusion center, and so the total number of measurements received at the fusion center is $m=RK$. 

Analogous to the tree topology, here we assume that all the $RK$ channels are independent Bernoulli erasure channels with probability of observability $p$. The measurement from the $i\nth$ sensor of a branch is observed at the fusion center if and only if the data are not missing in all the $K-i+1$ channels between the $i\nth$ sensor and the fusion center which occurs with probability $p^{K-i+1}$. Therefore, the probability of a measurement being observed at the fusion center depends on the position of the sensor in the branch. This property distinguishes the serial-star topology from the other two topologies (star and tree) discussed in \Cref{sec:star} and \Cref{sec:tree}, and  so the serial-star topology demands separate analysis.

Let the indices of the measurements which are observed at the fusion center be $\calT\subset\{1,2,\ldots,m\}$. The unknown sparse vector $\vecx$ can be recovered using the observed sensor measurements given by \eqref{eq:sys_effective}. The following result gives a sufficient condition to ensure the exact recovery of the sparse vector under the above model. 

\begin{theorem}\label{thm:ripline}
Consider a serial-star topological sensor network with $R$ branches and $K$ sensors per branch, whose measurement model is given by \eqref{eq:sys_star}. Each link in the network is modeled using an independent Bernoulli erasure channel with probability of observability $p\in[0,1]$. Suppose $\matA\in\bbR^{m\times N}$ is a subGaussian random matrix with parameter $c$. Let $\calT\subseteq\{0,1,\ldots,m\}$ be the set of indices of measurements observed at the fusion center. Then, for any $\epsilon>0$, if $R\geq \bline$ where
\begin{multline}\label{eq:m_line}
\bline\geq \ls\ln \frac{1-pe^{-C\delta^2}}{1-p+p(1-e^{-C\delta^2})(pe^{-C\delta^2})^K}\rs^{-1}\\
\times\ls \frac{4}{3}s\ln\lb \frac{eN}{s}\rb+\frac{14}{3}s+\frac{4}{3}\ln 2\epsilon^{-1}\rs
\end{multline}
the {RIC} $\delta_s$ of $\lv\calT\rv^{-1}\matA_{\calT}$ given in \eqref{eq:sys_effective} satisfies $\delta_s < \delta$ for all $0<\delta<1$ with probability at least $1-\epsilon$. Here, $C>0$ is a constant dependent only on $c$. 
\end{theorem}
\begin{proof}
See \Cref{app:ripline}.
\end{proof}

The above RIP-based result immediately implies a result similar to \Cref{cor:noisy} which ensures faithful recovery of a sparse vector when the bound in \Cref{thm:ripline}  is satisfied. We omit the statement of the corollary to avoid repetition. We next establish the monotonically decreasing nature of $\bline$ with the probability of observability $p$:
\begin{prop}\label{prop:ratioline}
Let $\bline$ and $\beta$ be as defined in \eqref{eq:m_line} and \eqref{eq:std_bound}. Then, the following statements hold:
\begin{enumerate}[label={(\roman*)}]
\item $\bline(p,\delta,K)\geq\bline(p,\sqrt{K}\delta)\geq \beta(\sqrt{K}\delta)/p$, for all $p\in[0,1]$, $\delta\in(0,1)$ and all values of $K$.
\item $\bline(p,\delta,K)$ is a strictly decreasing function of $p$ for a fixed value of $\delta$ and $K$.
\end{enumerate}
\end{prop}
\begin{proof}
See \Cref{app:ratioline}.
\end{proof}
Statement (i) of \Cref{prop:ratioline} implies that the bound on the required number of measurements is larger for the serial-star topology compared to the generalized star topology discussed in \Cref{cor:riptree}. This observation is similar to Statement (i) of \Cref{cor:riptree}, and is due the multi-hop setting of the serial-star topology. Also, Statement (ii) of \Cref{prop:ratioline} is intuitive as $p$ increases, the fusion center is more likely to receive the sufficient number of measurements with a smaller value of $R$. Consequently, the bound decreases.

The other key inferences from \Cref{thm:ripline} are as follows:
\subsection{Special cases}
We consider the following special cases:

\begin{itemize}
\item \emph{$p=1$: } In this case, the bound $\bline$ in \eqref{eq:m_line} simplifies as follows:
\begin{equation}
m=RK\geq \beta(\delta),
\end{equation}
where $\beta$ defined in \eqref{eq:std_bound} corresponds to the classical result  for the conventional sparse signal recovery problem with no missing data. 
This connection is obvious because when $p=1$, none of the measurements are missing, and so exactly $RK$ measurements are available at the fusion center.
 \item \emph{$p=0$ or $K=0$: } When either of $p$ or $K$ is zero, the bound gives $m\geq \infty$ because none of the measurements reach the fusion center. Therefore, the sparse signal  recovery fails with probability one.
\item \emph{$K=1$: } For this case, the bound in \eqref{eq:m_line} simplifies as follows:
\begin{equation}
R\geq \bstar(p,\delta),
\end{equation}
where $\bstar$ defined in \eqref{eq:m_bound} corresponds to the star topology. This relation is expected as  the system is equivalent to having all sensors directly connected to the fusion center via independent Bernoulli erasure channels with probability of observability $p$. 
\end{itemize} 

\subsection{Similarities with the bound for tree topology}
We observe the following similarities between the bound $\btree$ in \eqref{eq:m_tree} for the tree topology and the bound $\bline$ in \eqref{eq:m_line} for the serial-star topology:
\begin{itemize}

\item Similar to the tree topology case, the bound on $R$ becomes smaller when $K$ grows and the other parameters remain unchanged. This dependence is obvious from \eqref{eq:m_line} and the fact that $p e^{-C\delta^2}<1$, for all values of $p,\delta\in(0,1]$. This monotonic behavior of $\bline$ is expected because as $K$ increases, more number of measurements are likely be available at the fusion center, which leads a smaller bound on the number of branches. Therefore, a lower bound on $R$ which which is independent of $K$ is:
\begin{multline}\label{eq:lowerbound_branch}
R\geq \lim_{K\to\infty} \bline = \ls\ln \frac{1-pe^{-C\delta^2}}{1-p}\rs^{-1}\\
\times\ls \frac{4}{3}s\ln\lb \frac{eN}{s}\rb+\frac{14}{3}s+\frac{4}{3}\ln 2\epsilon^{-1}\rs.
\end{multline}

\item The bound $\bline$ in \eqref{eq:m_line} scales $\approx\calO(s\ln (N/s)/\delta^2)$ if we keep the other parameters constant. This statements also hold for tree  topology. Also, for both topologies, the corresponding bounds  strictly decrease as $p$ increases (see \Cref{prop:ratiotree,prop:ratioline}).

\end{itemize}
\section{Numerical Experiments}
\begin{figure}
\begin{center}
\includegraphics[width=8cm]{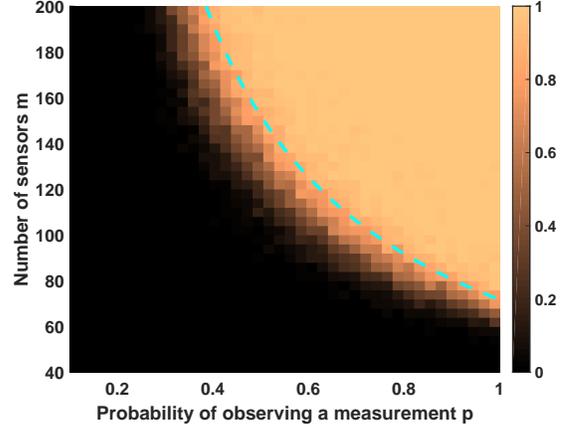}
\end{center}
\caption{Variation of the success probability of LASSO with probability of observability $p$ and the number of sensors for the star topology when $N=200$ and $s=20$. (Cyan dotted line represents the transition in the success probability). The figure corroborates \Cref{thm:rip,prop:ratio} by showing the monotonically increasing nature of the success probability with $p$ and $m$.}
\label{fig:starplot}
\end{figure}

\begin{figure}
\begin{center}
\subfloat[Varying $R$ with $K=10$ (Cyan dotted line represents the transition in the success probability)]
{\label{fig:treeR}
\includegraphics[width=8cm]{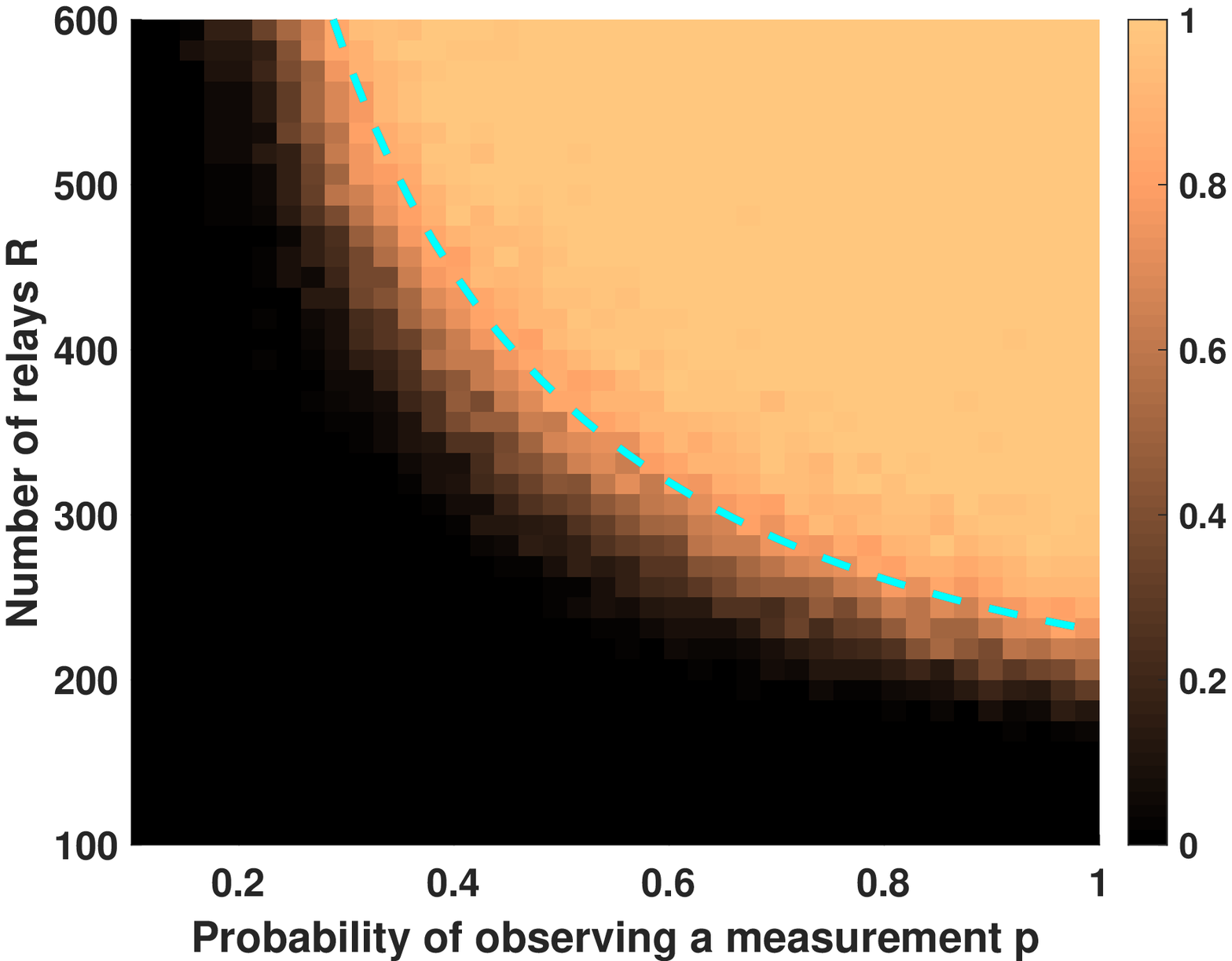}
}

\subfloat[Varying $K$ with $R=1$]
{\label{fig:treeK}
\includegraphics[width=8cm]{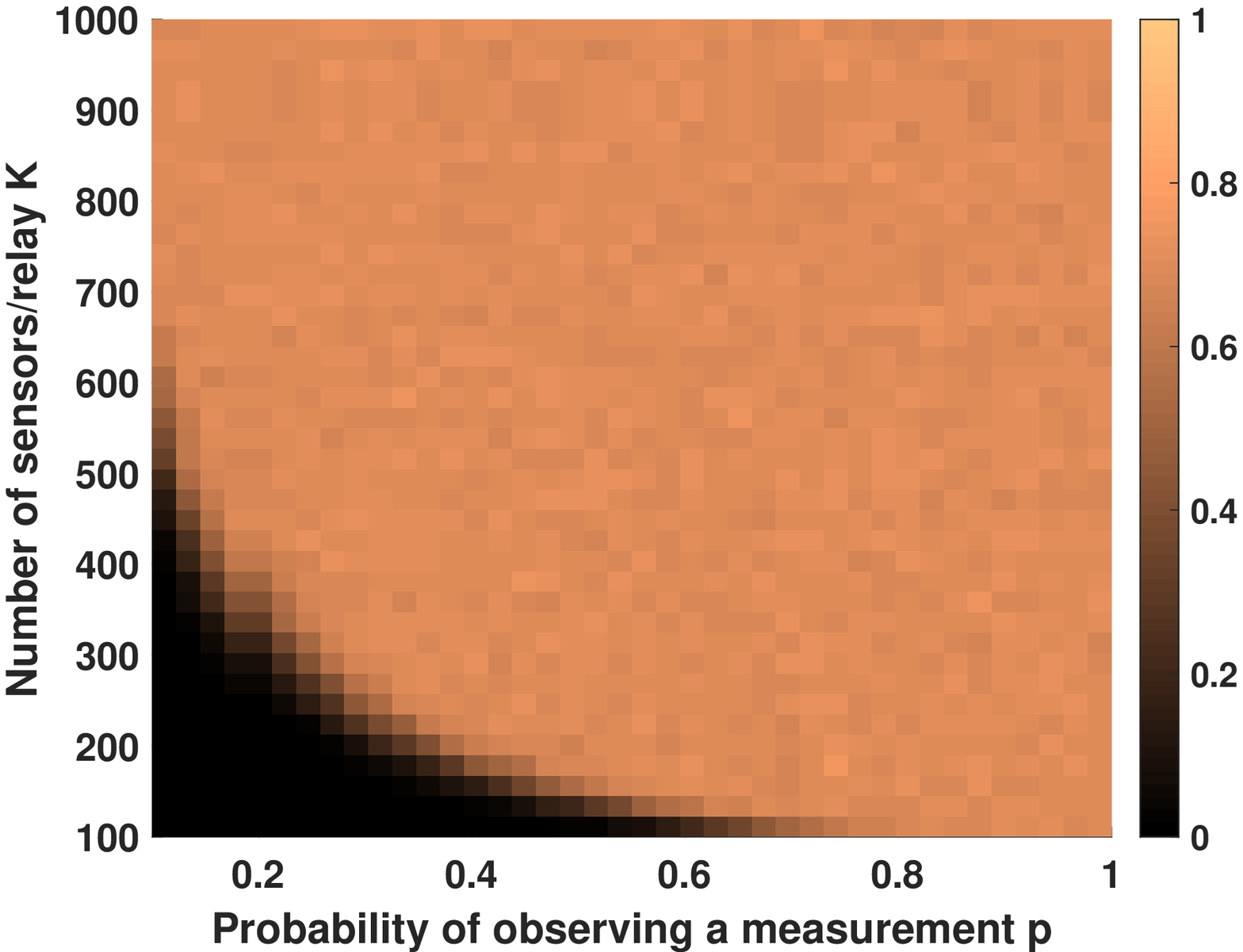}
}
\end{center}
\caption{Variation of the success probability of LASSO with probability of observability $p$ for the tree topology when $N=200,s=20$, and $q=0.7$. The strictly increasing nature of the success probability with $R$ in \Cref{fig:treeR} and the upper bound on the success probability with increasing $K$ in \Cref{fig:treeK} confirm the results in \Cref{thm:riptree,prop:ratiotree}.}
\label{fig:treeplot}
\end{figure}

\begin{figure}
\begin{center}
\subfloat[Varying $R$ with $K=10$ (Cyan dotted line represents the transition in the success probability)]
{\label{fig:lineR}
\includegraphics[width=8cm]{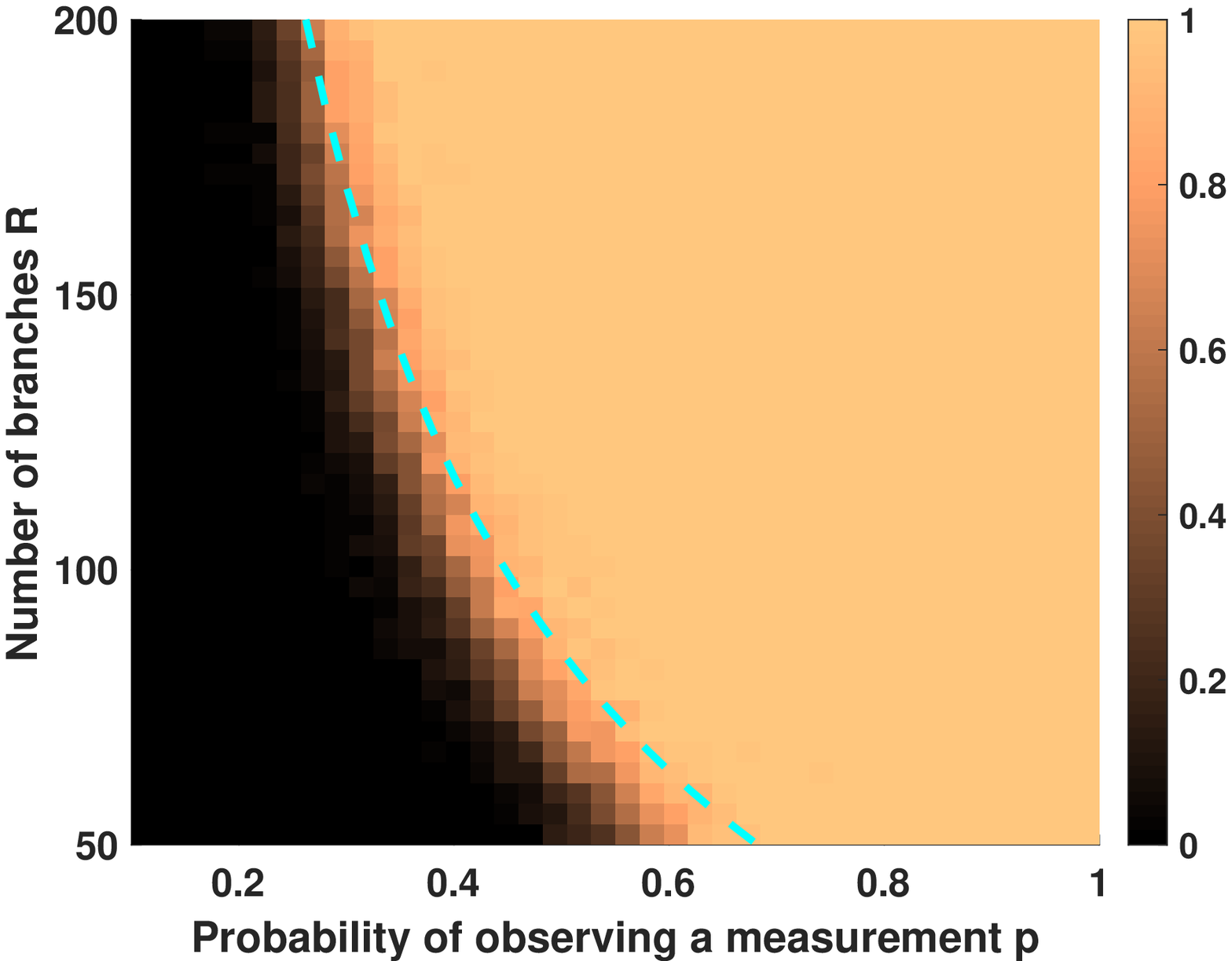}
}

\subfloat[Varying $K$ with $R=5$]
{\label{fig:lineK}
\includegraphics[width=8cm]{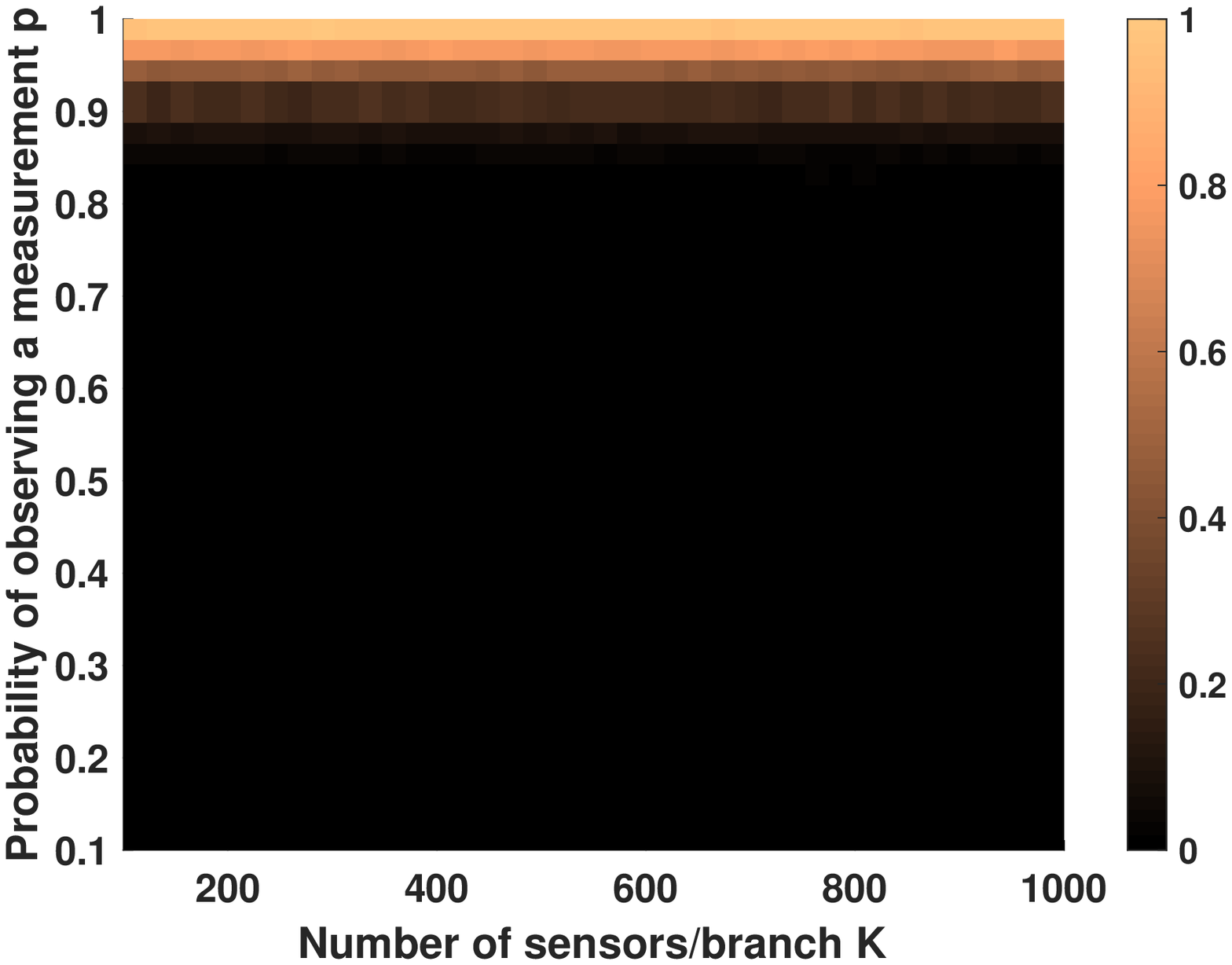}
}
\end{center}
\caption{Variation of the success probability of LASSO with probability of observability $p$ for a serial-star topology when $N=200$ and $s=20$. As the results in \Cref{thm:ripline,prop:ratioline} show, the success probability grows with $R$ while it goes to 1 when $K$ becomes larger for a fixed value of $R$ only when $p$ is close to 1.}
\label{fig:lineplot}
\end{figure}

The recovery performance of the CS algorithm for different sparsity levels and the number of measurements when there is missing data is thoroughly studied in \cite{charbiwala2010compressive,cao2016data} using numerical simulations.  Hence, we focus only on the variation of the success probability of recovery using LASSO with the probability of observability $p$, to corroborate our results in \Cref{sec:star,sec:tree,sec:line}. Our setting is as follows: the length of the unknown sparse vector is $N=200$, and the sparsity $s=20$. We declare that the recovery is successful if $\frac{\lV\vecx-\hat{\vecx}\rV}{\lV\vecx\rV}< 10^{-2}$. The average probability of this event is plotted as a color-map in \Cref{fig:starplot,fig:treeplot,fig:lineplot}. In the plots, the cyan dotted curve represents a function of $p$ shown on the colormap, which indicates the transition in the success probability curve. This function is obtained by fitting a curve to the minimum value on the Y-axis for which the success probability is greater than $0.9$.
\subsection{Star topology}
 The recovery performance for a star topological network with varying values of probability of observability $p$ and the number of sensors $m$ is shown in \Cref{fig:starplot}. As $p$ or $m$ increases, the success probability grows, verifying \Cref{prop:ratio}. Moreover, the value of $m$ for which there is a transition in the success probability approximately changes as $-\frac{\alpha}{\log {1-\lambda p}}$ (obtained by fitting a function of $p$ that is indicated by the  cyan dotted curve in \Cref{fig:lineplot}) with $\alpha=16$ and $\lambda=0.2$. This observation is in agreement with the term $\ls\ln\frac{1}{ 1-p+p\exp(-C\delta^2)}\rs^{-1}$ given in \Cref{thm:rip}.
\subsection{Tree topology} 
The success probability of sparse recovery for a tree topological network with varying values of probability of observability $p$ (for the channels between sensors and their corresponding relays), the number of relays $R$ and the number of sensors per relay $K$ is shown in \Cref{fig:treeplot}. We set $K=10$ for \Cref{fig:treeR}, and $R=1$ for \Cref{fig:treeK}. For these plots, we set the probability of observing measurement at the fusion center $q=0.7$.
\begin{itemize}
\item \Cref{fig:treeplot} shows that as $p$, $R$, or $K$ increases, the success probability grows. This inference validates Statement (ii) of \Cref{prop:ratiotree}.
\item From \Cref{fig:treeR}, the value of $R$ for which there is a transition in the success probability approximately changes as $-\frac{\alpha}{\log {1-q+q(1-\lambda p)^K}}$ (obtained by fitting a function of $p$ that is indicated by the  cyan dotted curve in \Cref{fig:treeR}) with $\alpha=225$ and $\lambda=0.2$. This observation is in agreement with the term $\ls\ln\frac{1}{ 1-q+q\lb 1-p+p e^{-C\delta^2}\rb^{K}}\rs^{-1}$ given in \Cref{thm:riptree}. 
\item \Cref{fig:treeK} indicates that the success probability does not approach 1 if we increase $K$ while keeping $R$ a constant. This observation confirms the lower bound on $R$ given by \eqref{eq:lowerbound_relay}.
\item \Cref{fig:starplot} and \Cref{fig:treeR} show that for the same value of $p$, high success probability (close to 1) is obtained when the number of sensors ($m=RK$) is much larger for the tree topology, compared to the same for the star topology. This is expected because a measurement reaches the fusion center only if both the channel between the sensor and the corresponding relay, and the relay and the fusion center do not erase the data. Hence, the success probability is smaller when sensors are arranged according to a tree topology compared to the star topology. 
\end{itemize}
\subsection{Serial-star topology}
The recovery performance for a serial-star topological network with varying values of probability of observability $p$, the number of branches $R$ and the number of sensors per branch $K$ is shown in \Cref{fig:lineplot}. We set $K=10$ for \Cref{fig:lineR}, and $R=5$ for \Cref{fig:lineK}. 
\begin{itemize}
\item Similar to the other plots,  \Cref{fig:lineplot} indicates that as $p$, $R$, or $K$ increases, the success probability grows. This inference is in agreement with \Cref{prop:ratioline}.
\item From \Cref{fig:lineR}, the value of $R$ for which there is a transition in the success probability approximately varies as $\alpha\ls\ln \frac{\alpha(1-p\lambda )}{1-p+p(1-\lambda)(p \lambda )^K}\rs^{-1}$ (obtained by fitting a function of $p$ that is indicated by the  cyan dotted curve in \Cref{fig:lineR}) with $\alpha=50$ and $\lambda=0.8$. This observation is in agreement with the term $\ls\ln \frac{1-pe^{-C\delta^2}}{1-p+p(1-e^{-C\delta^2})(pe^{-C\delta^2})^K}\rs^{-1}$ given in \Cref{thm:ripline}. 
\item \Cref{fig:lineK} shows that for a fixed value of $R$ the success probability approaches 1, only if $p$ is close to 1, irrespective of $K$. This observation confirms the lower bound on $R$ given by \eqref{eq:lowerbound_branch}. A comparison of \Cref{fig:treeK,fig:lineK} shows that the behavior of the success probability for the tree topology is different from that of the serial-star topology and the success probability for the tree topology does not approach 1 even when $p=1$. This is intuitive because the channel between the relays and the fusion center erases some data even if $p=1$ (note that $q=0.7$). So, as we mentioned in \Cref{sub:nature}, the success probability is upper bounded by $1-(1-q)^R=0.7$. However, for the serial-star topology,  the success probability is bounded by $1-(1-p)^R$ which goes to $1$ as $p$ approaches 1. Also, it is worth noting that the tree topology with $q=1$ corresponds to the star topology whose recovery performance is plotted in \Cref{fig:starplot}.
\item Figures \ref{fig:starplot} and \ref{fig:lineR} show that for the same value of $p$, high success probability (close to 1) is obtained when the number of sensors ($m=RK$) is much larger for the tree topology, compared to the same for the star topology. This is expected because a measurement reaches the fusion center only if none of the channels between the corresponding sensor and the fusion center erase the data. 
\end{itemize}
\section{Conclusion}
We considered the problem of recovering a sparse signal from its noisy linear measurements under missing data in a sensor network. The missing data mechanism was modeled using an independent Bernoulli process parameterized by the probability of observability associated with  each channel in the network. We derived a sufficient condition required to faithfully recover the unknown signal using standard CS algorithms, for the star, tree, and serial-star topologies. Our  RIP-based analysis established that the additional measurements required depend only on the probability of observability, the RIC of the measurement matrix, and the topology. 

It would be interesting to  extend our results to the other sparsity patterns like block sparsity, piece-wise sparsity, etc. Also, a similar analysis for quantized sparse recovery problems like 1-bit CS is another direction for future work.  Finally, we assumed the knowledge of the missing indices of measurements at the fusion center. Devising algorithms for the setting where the missing indices are unknown, and obtaining their theoretical guarantees can be an interesting avenue for future work as well.
\appendices
\crefalias{section}{appendix}
\section{Proof of \Cref{thm:rip}}\label{app:rip}
\subsection{Toolbox}
In this subsection, we list some of the existing results that are used to prove \Cref{thm:rip}. 
\begin{lemma}\label{lem:RIP}
Let $\matA$ be an $m\times N$ matrix and $\delta_s$ be $s-$RIC of $\matA$. Then, for any $\rho>0$,
\begin{multline}
\bbP\lc\delta_s>\delta\rc\leq \binom{N}{s}\lb1+\frac{2}{\rho}\rb^s\\\times\bbP\lc \vecu\in\bbR^N:\lv \lV\matA\vecu\rV-\lV\vecu\rV\rv\geq (1-2\rho)\delta\lV\vecu\rV\rc,
\end{multline}
\end{lemma}
\begin{proof}
The proof directly follows from the proof of \cite[Theorem 9.11]{foucart2013mathematical}.
\end{proof}

\begin{lemma}[{\cite[Lemma 9.8]{foucart2013mathematical}}]\label{lem:prob_subGaussian}
Let $\matA$ be an $m\times N$ subGaussian random matrix with parameter $c$ (See \Cref{def:subGaussian}). Then, for all $\vecz\in\bbR^N$, and every $t\in(0,1)$, the following relation holds:
\begin{equation}
\bbP\lc \lv m^{-1}\lV\matA\vecz\rV^2-\lV\vecz\rV^2\rv\geq t\lV\vecz\rV\rc\leq 2 \exp\lb-C t^2m\rb,
\end{equation}
where $C$ is a constant that depends only on $c$.
\end{lemma}
\subsection{Proof of \Cref{thm:rip}}
Let $\bar{\matA}\triangleq \lv\calT\rv^{-1}\matA_{\calT}\in\bbR^{\lv\calT\rv\times N}$. From \Cref{lem:RIP}, for any $\rho>0$,
\begin{multline}
\bbP\lc\delta_s>\delta\rc
\leq \binom{N}{s}\lb1+\frac{2}{\rho}\rb^s\\ \times\bbP\lc \vecu\in\bbR^N:\lv \lV\bar{\matA}\vecu\rV-\lV\vecu\rV\rv\geq t\lV\vecu\rV\rc,\label{eq:prob_star}
\end{multline}
where we define $t\triangleq(1-2\rho)\delta$. Here, to compute the probability term, we need to handle the randomness in $\bar{\matA}$ which arises due to the randomness in both $\calT$ and $\matA$. 

Further, since $\calT$ is the set of  indices corresponding to the available data, $\lv\calT\rv$ is the sum of $m$ independent  Bernoulli variables. Thus, $\lv\calT\rv$ follows a binomial distribution and 
\begin{equation}\label{eq:binom}
\bbP\lc\lv\calT\rv=i\rc=\binom{m}{i}(1-p)^{m-i}p^i.
\end{equation} 
Therefore, for any $\vecz\in\bbR^{N}$ and $t>0$,
\begin{align}
\bbP\lc \lv\lV\bar{\matA}\vecz\rV^2-\lV\vecz\rV^2\rv\geq t\lV\vecz\rV\rc&\notag\\
&\hspace{-4cm} = \bbP\lc\lv\calT\rv=0\rc + \sum_{i=1}^{m}\bigg[ \bbP\lc\lv\calT\rv=i\rc\\
&\hspace{-2.75cm}\times\bbP\lc \lv \lV\bar{\matA}\vecz\rV^2-\lV\vecz\rV^2\rv\geq t\lV\vecz\rV\middle|\lv\calT\rv=i\rc\bigg]\\
&\hspace{-4cm}\leq (1-p)^m + \sum_{i=1}^m \binom{m}{i}(1-p)^{m-i}p^i \times 2\exp\lb-Ct^2i\rb\label{eq:boundAz_rip}\\
&\hspace{-4cm}\leq 2\sum_{i=0}^m \binom{m}{i}(1-p)^{m-i}\lb p e^{-Ct^2}\rb^i\\
&\hspace{-4cm} = 2\ls1-p+p e^{-Ct^2}\rs^m.
\end{align}
where we use \eqref{eq:binom} and \Cref{lem:prob_subGaussian} to obtain \eqref{eq:boundAz_rip}.  

Finally, from \Cref{lem:RIP}, Stirling's approximation~{\cite[Lemma C.5]{foucart2013mathematical}}, and the definition of $t$, we conclude from \eqref{eq:prob_star} that
\begin{multline*}
\bbP(\delta_s>\delta)
\leq 2\lb\frac{eN}{s}\rb^s\lb1+\frac{2}{\rho}\rb^s\\\times\lb 1-p+p\exp(-C\delta^2(1-2\rho)\rb^m.
\end{multline*}
Setting $\rho=2/(e^{7/2}-1)$ and using $1-2\rho>\frac{\sqrt{3}}{2}$, we get $\delta_s<\delta$ with probability at least $1-\epsilon$ if $m\geq\bstar$(defined in the statement of the theorem). Thus, the proof is complete.
\hfill\qed
\section{Proof of \Cref{prop:ratio}}\label{app:ratio}
To prove Statement (i), we need to show that
%\begin{equation}
% \ls\ln\lb\frac{1}{1-p+pe^{-C\delta^2}}\rb\rs^{-1}\geq \frac{1}{pC\delta^2},
%\end{equation}
%which is equivalent to showing that
\begin{equation}\label{eq:ratiointer_1}
f_1(p,\delta)\triangleq1-p+pe^{-C\delta^2}-\exp\lb-pC\delta^2\rb\geq 0.
\end{equation}
We have 
\begin{align}
\frac{\partial f_1}{\partial p} &= -1+\exp(-C\delta^2)+C\delta^2 \exp\lb-pC\delta^2\rb\\
&\begin{cases}
> 0 & \text{ for } p < \tilde{p}\\
=0 & \text{ for } p = \tilde{p}\\
<0 & \text{ for }  p>\tilde{p},
\end{cases}
\end{align} 
where $\tilde{p}=\frac{1}{C\delta^2}\ln\lb\frac{1-\exp(-C\delta^2)}{C\delta^2}\rb$. Thus, for $p\in[0,1]$, $
f_1(p,\delta)\geq \min\lc f_1(0,\delta),f_1(1,\delta)\rc=0$. Therefore, \eqref{eq:ratiointer_1} holds, and the proof for Statement (i) is complete.
 
Next, to prove Statement (ii), we define
\begin{equation}
f_2(p,\delta)=1-p+p\exp(-C\delta^2).
\end{equation}
Then, we obtain that
\begin{equation}
\frac{\partial f_2}{\partial p} = -1+\exp(-C\delta^2)< 0,
\end{equation} 
since $0<\exp(-C\delta^2)<1$, for all values of $\delta$.
Thus, $f_2$ is a decreasing function of $p$. Further, 
\begin{equation}
\bstar(p,\delta) = -\lb \ln f_2(p,\delta)\rb^{-1}.
\end{equation}
The above relation implies that $\bstar$ is an increasing function of $f_2(p,\delta)$ and hence, $\bstar$ is also a strictly decreasing function of $p$. Hence, the proof is complete. \hfill\qed

\section{Proof of \Cref{thm:riptree}}\label{app:riptree}
To prove \Cref{thm:riptree}, we define the following polynomial:
\begin{equation}\label{eq:gdefn}
g(x) \triangleq \ls 1-q+q(1-p+px)^{K}\rs^R = \sum_{i=0}^{m} D_i x^i,
\end{equation}
where $m =RK$ and $D_i$ is the co-efficient of $x^i$. We note that $D_i\geq 0$ and $\sum D_i=g(1)=1$. Therefore, $\lc D_i\rc$ is a valid probability mass function.
\begin{lemma}\label{lem:generatortree}
The probability $\bbP\lc\lv\calT\rv=i\rc = D_i$ where $D_i$ is as defined in \eqref{eq:gdefn}.
\end{lemma}
\begin{proof}
The proof relies on the technique of generating functions that is widely used in discrete mathematics~\cite{wilf2005generatingfunctionology}. To compute $\bbP\lc\lv\calT\rv=i\rc$, we need to consider all nonnegative integer solutions of the following problem:
\begin{equation}\label{eq:integereq}
\sum_{r=1}^R k_r= i \text{ subject to } 0\leq k_r\leq K.
\end{equation}
Here, $k_r$ represents the number of sensor measurements available at the $r\nth$ relay node. Further, the probability of $k_r$ sensor measurements being available at the $r\nth$ relay node is $\binom{K}{k_r}p^k(1-p)^{K-k_r}$. Therefore, accounting for the uncertainty of data from the relay to the fusion center, the probability $p_r(k_r)$ of receiving $k_r$ sensor measurements from the $r\nth$ sensor is 
\begin{equation}
p_r(k_r) = \begin{cases}
q\binom{K}{k}p^k(1-p)^{K-k} &\text{ if } k_r\neq 0\\
1-q + q (1-p)^{K} &\text{ if } k_r = 0.
\end{cases}
\end{equation}
Thus, for each relay, we construct a  factor of generating function as follows:
\begin{multline}
g'(x) = \prod_{r=1}^R\Bigg[ \lb 1-q+q(1-p)^{K} \rb x^0\\
+  \sum_{k=1}^{K} q\binom{K}{k}p^k(1-p)^{K-k}x^k \Bigg].
\end{multline}
This construction ensures that the coefficient of $x^i$ is the probability of $i$ measurements being available at the fusion center. Finally, on simplifying $g'(x)$, we get that $g'(x)=g(x)$ and the proof is complete.
\end{proof}
\subsection{Proof of \Cref{thm:riptree}}
Let $\bar{\matA}\triangleq \lv\calT\rv^{-1}\matA_{\calT}\in\bbR^{\lv\calT\rv\times N}$. From \Cref{lem:RIP}, for any $\rho>0$,
\begin{multline}
\bbP\lc\delta_s>\delta\rc
\leq \binom{N}{s}\lb1+\frac{2}{\rho}\rb^s\\\times\bbP\lc \vecu\in\bbR^N:\lv \lV\bar{\matA}\vecu\rV-\lV\vecu\rV\rv\geq t\lV\vecu\rV\rc,\label{eq:prob_tree}
\end{multline}
where we define $t\triangleq(1-2\rho)\delta$. Here, to compute the probability term, we need to handle the randomness in $\bar{\matA}$ which arises due to the randomness in both $\calT$ and $\matA$. 
Then, for any $\vecz\in\bbR^{N}$ and $t>0$,
\begin{align}
\bbP\lc \lv\lV\bar{\matA}\vecz\rV^2-\lV\vecz\rV^2\rv\geq t\lV\vecz\rV\rc&\notag\\
&\hspace{-4cm} = \bbP\lc\lv\calT\rv=0\rc + \sum_{i=1}^{m}\bigg[ \bbP\lc\lv\calT\rv=i\rc\notag\\
&\hspace{-2.75cm}\times\bbP\lc \lv \lV\bar{\matA}\vecz\rV^2-\lV\vecz\rV^2\rv\geq t\lV\vecz\rV\middle|\lv\calT\rv=i\rc\bigg]\\
&\hspace{-4cm}\leq D_0 + \sum_{i=1}^{m} D_i \times 2\exp\lb-Ct^2i\rb\label{eq:boundAz}\\
&\hspace{-4cm}\leq 2g\lb e^{-Ct^2}\rb = 2\ls 1-q+q\lb 1-p+p e^{-Ct^2}\rb^{K}\rs^R,\label{eq:boundAz1}
\end{align}
where we use \Cref{lem:generatortree} and \Cref{lem:prob_subGaussian} to obtain \eqref{eq:boundAz}. Also, we use \eqref{eq:gdefn} to get \eqref{eq:boundAz1}.

Finally, from \Cref{lem:RIP}, Stirling's approximation~{\cite[Lemma C.5]{foucart2013mathematical}}, and the definition of $t$, we conclude from \eqref{eq:prob_tree} that
\begin{multline}
\bbP(\delta_s>\delta)
\leq 2\lb\frac{eN}{s}\rb^s\lb1+\frac{2}{\rho}\rb^s\\\times\ls 1-q+q\lb 1-p+p e^{-Ct^2}\rb^{K}\rs^R.
\end{multline}
Setting $\rho=2/(e^{7/2}-1)$ and using $1-2\rho>\frac{\sqrt{3}}{2}$, we get $\delta_s<\delta$ with probability at least $1-\epsilon$ if $R\geq \btree$ (defined in the statement of the theorem).
Thus, the proof is complete.
\hfill\qed

\section{Proof of \Cref{thm:ripline}}\label{app:ripline}
To prove \Cref{thm:riptree}, we define the following polynomial:
\begin{equation}\label{eq:hdefn}
h(x) \triangleq \ls \frac{1-p+p^{K+1}x^K(1-x)}{1-px}\rs^R = \sum_{i=0}^{m} d_i x^i,
\end{equation}
where $m = \sum_{r=1}^{R}K$ and $d_i$ is the co-efficient of $x^i$. We note that $d_i\geq 0$ and $\sum d_i=h(1)=1$. Therefore, $\lc d_i\rc$ is a valid probability mass function.
\begin{lemma}\label{lem:generatorline}
The probability $\bbP\lc\lv\calT\rv=i\rc = d_i$ where $d_i$ is as defined in \eqref{eq:hdefn}.
\end{lemma}
\begin{proof}
The proof relies on the technique of generating functions used in \Cref{app:riptree}. To compute $\bbP\lc\lv\calT\rv=i\rc$, we need to consider all nonnegative integer solutions of \eqref{eq:integereq} where $k_r$ represents the number of sensor measurements available to the fusion center from the $r\nth$ branch. Further, the probability $p_r(k_r)$ of $k_r$ sensor measurements being available at the $r\nth$ branch is
\begin{equation}
p_r(k_r) = \begin{cases}
(1-p)p^{k_r} &\text{ if } k_r\neq K\\
p^{k_r} &\text{ if } k_r = K.
\end{cases}
\end{equation}
Thus, for each relay, we construct a  factor of generating function as follows:
\begin{equation}
h'(x) = \prod_{r=1}^R\Bigg[ p^{K}x^K+  \sum_{k=0}^{K-1} (1-p)p^kx^k \Bigg].
\end{equation}
This construction ensures that the coefficient of $x^i$ is the probability of $i$ measurements being available at the fusion center. Finally, on simplifying $h'(x)$, we get
\begin{align}
h'(x) &= \ls p^Kx^K+(1-p)\frac{1-(px)^K}{1-px}\rs^R\\
&= \ls \frac{1-p+p^{K+1}x^{K}(1-x)}{1-px}\rs^R=h(x).
\end{align}
Hence, the proof is complete.
\end{proof}
\subsection{Proof of \Cref{thm:ripline}}
Let $\bar{\matA}\triangleq \lv\calT\rv^{-1}\matA_{\calT}\in\bbR^{\lv\calT\rv\times N}$. From \Cref{lem:RIP}, for any $\rho>0$,
\begin{multline}
\bbP\lc\delta_s>\delta\rc
\leq \binom{N}{s}\lb1+\frac{2}{\rho}\rb^s\\\times\bbP\lc \vecu\in\bbR^N:\lv \lV\bar{\matA}\vecu\rV-\lV\vecu\rV\rv\geq t\lV\vecu\rV\rc,\label{eq:prob_line}
\end{multline}
where we define $t\triangleq(1-2\rho)\delta$. Here, to compute the probability term, we need to handle the randomness in $\bar{\matA}$ which arises due to the randomness in both $\calT$ and $\matA$. 

Therefore, for any $\vecz\in\bbR^{N}$ and $t>0$,
\begin{align}
\bbP\lc \lv\lV\bar{\matA}\vecz\rV^2-\lV\vecz\rV^2\rv\geq t\lV\vecz\rV\rc&\notag\\
&\hspace{-4.2cm} = \bbP\lc\lv\calT\rv=0\rc + \sum_{i=1}^{m}\bigg[ \bbP\lc\lv\calT\rv=i\rc\notag\\
&\hspace{-2.75cm}\times\bbP\lc \lv \lV\bar{\matA}\vecz\rV^2-\lV\vecz\rV^2\rv\geq t\lV\vecz\rV\middle|\lv\calT\rv=i\rc\bigg]\\
&\hspace{-4.2cm}\leq d_0 + \sum_{i=1}^{m} d_i \times 2\exp\lb-Ct^2i\rb\label{eq:boundAz_line}\\
&\hspace{-4.2cm}\leq 2h\lb e^{-Ct^2}\rb = 2\ls \frac{1-p+p^{K+1}e^{-Ct^2K}(1-e^{-Ct^2})}{1-pe^{-Ct^2}}\rs^R,\label{eq:boundAz1_line}
\end{align}
where we use \Cref{lem:generatorline} and \Cref{lem:prob_subGaussian} to obtain \eqref{eq:boundAz_line}. Also, we use \eqref{eq:hdefn} to get \eqref{eq:boundAz1_line}.

Finally, from \Cref{lem:RIP}, Stirling's approximation~{\cite[Lemma C.5]{foucart2013mathematical}}, and the definition of $t$, we conclude from \eqref{eq:prob_line} that
\begin{multline}
\bbP(\delta_s>\delta)
\leq 2\lb\frac{eN}{s}\rb^s\lb1+\frac{2}{\rho}\rb^s\\\times\ls  \frac{1-p+p^{K+1}e^{-Ct^2K}(1-e^{-Ct^2})}{1-pe^{-Ct^2}}\rs^R.
\end{multline}
Setting $\rho=2/(e^{7/2}-1)$ and using $1-2\rho>\frac{\sqrt{3}}{2}$, we get $\delta_s<\delta$ with probability at least $1-\epsilon$ if $R\geq \bline$(defined in the statement of the theorem).
Thus, the proof is complete.
\hfill\qed

\section{Proof of \Cref{prop:ratioline}}\label{app:ratioline}
In the following subsections, we prove the two statements of  \Cref{prop:ratioline}.
\subsection{Proof of Statement (i)}
To prove Statement (i), we need to show the following:
\begin{equation}\label{eq:firststep_ratioline}
\bline(p,\delta,K)\geq\bstar(p,\sqrt{K}\delta).
\end{equation}
Then, the remaining part of Statement (i) follows from \Cref{prop:ratio}. 
%From \eqref{eq:m_bound} and \eqref{eq:m_line}, \eqref{eq:firststep_ratioline} is proved if the following holds:
%\begin{multline}
%\ls\ln \frac{1-pe^{-C\delta^2}}{1-p+p(1-e^{-C\delta^2})(pe^{-C\delta^2})^K}\rs^{-1}\\
%\geq \ls\ln\frac{1}{ 1-p+p\exp(-C\delta^2)}\rs^{-1}.
%\end{multline}
On further simplification, it is easy to see that the above relation is proved if we establish the following:
\begin{equation}\label{eq:secondstep_ratioline}
f_1(p,d)\triangleq 1-p -  d^{K-1}(1-p^K)+d^{K}(p -p^K)\geq 0,
\end{equation}
for all $p\in[0,1]$ and $d\in(0,1)$ where $d=e^{-C\delta^2}$. Then, 
\begin{equation}\label{eq:f1_derivative}
\frac{\partial f_1}{\partial d}=d^{K-2}\lb Kd(p -p^K)-(K-1)(1-p^K)\rb<0,
\end{equation}
for all $p\in[0,1]$, if $0<d<\frac{(K-1)(1-p^K)}{K(p -p^K)}$. 

Next, we show that $\frac{(K-1)(1-p^K)}{K(p -p^K)}\geq 1$ which implies that $\frac{\partial f_1}{\partial d}<0$ for $0<d<1$. For this, we define 
\begin{equation}\label{eq:f2_defn}
f_2(p)\triangleq (K-1)(1-p^K)-K(p -p^K)= K-1-Kp+p^K,
\end{equation}
and show that $f_2(p)\geq 0$. Then, for all $p\in[0,1]$,
\begin{equation}
\frac{\partial f_2}{\partial p} = -K(1-p^{K-1})\leq 0.
\end{equation}
Therefore, $f_2(p)\geq f_2(1)=0$, for all $p\in[0,1]$. Hence, \eqref{eq:f2_defn} implies that $\frac{(K-1)(1-p^K)}{K(p -p^K)}\geq1$. Thus, \eqref{eq:f1_derivative} gives $\frac{\partial f_1}{\partial d}<0$ for $0<d<1$. Consequently, we obtain that, for all $p\in[0,1]$ and $0<d<1$,
\begin{equation}
f_1(p,d) \geq f_1(p,1) =0.
\end{equation} 
Thus, from \eqref{eq:secondstep_ratioline}, the proof is complete.

\subsection{Proof of Statement (ii)} 
To prove Statement (ii), we define the function $f_3$ as follows:
\begin{equation}
f_3(p)\triangleq  \frac{1-pd}{1-p+p(1-d)(pd)^K},
\end{equation}
where $d=e^{-C\delta^2}$.
We note that 
\begin{equation}\label{eq:ratioline_inter2}
\bline =C\delta^2\ln^{-1} (f_3(p))\beta(\delta),
\end{equation}
where $\beta(\delta)$ is defined in \eqref{eq:std_bound}. Therefore, to show that $\bline$ is a decreasing function of $p$, it suffices to show that $f_3(p)$ is an increasing function of $p$. We prove this by showing that the derivative $\frac{\partial f_3}{\partial p}>0$. We have
\begin{multline}
\frac{\partial f_3}{\partial p}=\frac{1-d}{\ls 1-p+p(1-d)(pd)^K\rs^2}\\
\times \lb 1-(K+1)(pd)^K +K(pd)^{K+1}\rb.
\end{multline}
Since $d=e^{-C\delta^2}<1$, we deduce that
\begin{equation}
\frac{1-d}{\ls 1-p+p(1-d)(pd)^K\rs^2}>0,
\end{equation}
for all values of $p\in[0,1]$ and $\delta\in(0,1)$. Thus, we obtain that $\frac{\partial f_3}{\partial p}>0$ if and only $f_4(pd)>0$ where we define the function $f_4$ as
\begin{equation}
f_4(x) = 1-(K+1)x^K +Kx^{K+1}.
\end{equation}
Here, the derivative $\frac{\partial f_4}{\partial x}$ is as follows:
\begin{equation}
\frac{\partial f_4}{\partial x} = -K(K+1)x^{K-1}(1-x)
\leq 0,
\end{equation}
if $0\leq x<1$.  Also, $0\leq pd=pe^{-C\delta^2}<1$ for all values of $p\in[0,1]$ and $\delta\in(0,1)$. Therefore, we conclude that
\begin{equation}\label{eq:ratioline_inter1}
f_4(pd) >  \underset{x\in[0,1)}{\sup} f_4(x) =\lim_{x\to 1^{-}}f_4(x)= 0,
\end{equation}
Hence, we get that $f_4(pd)>0$ which, in turn, implies that $\frac{\partial f_3}{\partial p}>0$. Therefore, $f_3(p)$ is a strictly increasing function of $p$, and thus, \eqref{eq:ratioline_inter2} implies that $\bline$ is a strictly decreasing function of $p$. Hence, the proof is complete.
\hfill\qed
\bibliographystyle{IEEEtran}
\bibliography{Supporting_Files/IEEEabrv,Supporting_Files/bibJournalList,Supporting_Files/Missing_cite}
\end{document}